\documentclass[12pt,english]{article}
\usepackage[T1]{fontenc}
\usepackage[latin9]{inputenc}
\usepackage{geometry}
\usepackage{color}
\usepackage{verbatim}
\usepackage{amsmath}
\usepackage{amsthm}
\usepackage{amssymb}
\usepackage{graphicx}
\usepackage{setspace}
\usepackage[authoryear]{natbib}
\PassOptionsToPackage{normalem}{ulem}
\usepackage{ulem}
\makeatletter

\providecommand{\tabularnewline}{\\}

\theoremstyle{plain}
\newtheorem{assumption}{\protect\assumptionname}
\theoremstyle{plain}
\newtheorem{lem}{\protect\lemmaname}
\theoremstyle{remark}
\newtheorem{rem}{\protect\remarkname}
\theoremstyle{plain}
\newtheorem{thm}{\protect\theoremname}
\theoremstyle{definition}
 \newtheorem{example}{\protect\examplename}
\theoremstyle{plain}
\newtheorem{cor}{\protect\corollaryname}
\theoremstyle{definition}
\newtheorem{defn}{\protect\definitionname}
\theoremstyle{plain}
\newtheorem{lyxalgorithm}{\protect\algorithmname}

\usepackage{hyperref}
\hypersetup{
     colorlinks   = true,
     citecolor    = blue
}
\allowdisplaybreaks
\usepackage{bbm}

\makeatother

\usepackage{babel}
\providecommand{\algorithmname}{Algorithm}
\providecommand{\assumptionname}{Assumption}
\providecommand{\corollaryname}{Corollary}
\providecommand{\definitionname}{Definition}
\providecommand{\examplename}{Example}
\providecommand{\lemmaname}{Lemma}
\providecommand{\remarkname}{Remark}
\providecommand{\theoremname}{Theorem}

\begin{document}
\global\long\def\Acal{\mathcal{A}}%

\global\long\def\Bcal{\mathcal{B}}%

\global\long\def\Ccal{\mathcal{C}}%

\global\long\def\Dcal{\mathcal{D}}%

\global\long\def\hD{\hat{D}}%

\global\long\def\te{\tilde{e}}%

\global\long\def\tF{\tilde{F}}%

\global\long\def\bF{\bar{F}}%

\global\long\def\bG{\bar{G}}%

\global\long\def\hG{\hat{G}}%

\global\long\def\tG{\tilde{G}}%

\global\long\def\Hcal{\mathcal{H}}%

\global\long\def\tH{\tilde{H}}%

\global\long\def\bH{\bar{H}}%

\global\long\def\cH{\check{H}}%

\global\long\def\hL{\hat{L}}%

\global\long\def\Lcal{\mathcal{L}}%

\global\long\def\Mcal{\mathcal{M}}%

\global\long\def\Ncal{\mathcal{N}}%

\global\long\def\Qcal{\mathcal{Q}}%

\global\long\def\tM{\tilde{M}}%

\global\long\def\bM{\bar{M}}%

\global\long\def\hM{\hat{M}}%

\global\long\def\hr{\hat{r}}%

\global\long\def\tR{\tilde{R}}%

\global\long\def\bR{\bar{R}}%

\global\long\def\RR{\mathbb{R}}%

\global\long\def\SS{\mathbb{S}}%

\global\long\def\ts{\tilde{s}}%

\global\long\def\Tcal{\mathcal{T}}%

\global\long\def\hU{\hat{U}}%

\global\long\def\tu{\tilde{u}}%

\global\long\def\hV{\hat{V}}%

\global\long\def\bW{\bar{W}}%

\global\long\def\tW{\tilde{W}}%

\global\long\def\hw{\hat{w}}%

\global\long\def\tZ{\tilde{Z}}%

\global\long\def\agmin{\arg\min}%

\global\long\def\hbeta{\hat{\beta}}%

\global\long\def\tbeta{\tilde{\beta}}%

\global\long\def\bbeta{\bar{\beta}}%

\global\long\def\teps{\tilde{\varepsilon}}%

\global\long\def\heps{\hat{\varepsilon}}%

\global\long\def\hpi{\hat{\pi}}%

\global\long\def\hxi{\hat{\xi}}%

\global\long\def\txi{\tilde{\xi}}%

\global\long\def\tDelta{\tilde{\Delta}}%

\global\long\def\tdelta{\tilde{\delta}}%

\global\long\def\hdelta{\hat{\delta}}%

\global\long\def\bkappa{\bar{\kappa}}%

\global\long\def\hgamma{\hat{\gamma}}%

\global\long\def\hGamma{\hat{\Gamma}}%

\global\long\def\tGamma{\tilde{\Gamma}}%

\global\long\def\bGamma{\bar{\Gamma}}%

\global\long\def\hSigma{\hat{\Sigma}}%

\global\long\def\bSigma{\bar{\Sigma}}%

\global\long\def\tSigma{\tilde{\Sigma}}%

\global\long\def\hsigma{\hat{\sigma}}%

\global\long\def\hTheta{\hat{\Theta}}%

\global\long\def\dTheta{\dot{\Theta}}%

\global\long\def\bTheta{\bar{\Theta}}%

\global\long\def\cTheta{\check{\Theta}}%

\global\long\def\rTheta{\mathring{\Theta}}%

\global\long\def\hPsi{\hat{\Psi}}%

\global\long\def\hpsi{\hat{\psi}}%

\global\long\def\hPhi{\hat{\Phi}}%

\global\long\def\hphi{\hat{\phi}}%

\global\long\def\tXi{\tilde{\Xi}}%

\global\long\def\lnorm{\left\Vert \right\Vert }%

\global\long\def\abs{\left|\right|}%

\global\long\def\ttheta{\tilde{\theta}}%

\global\long\def\htheta{\hat{\theta}}%

\global\long\def\dv{\dot{v}}%

\global\long\def\rank{{\rm rank}\,}%

\global\long\def\trace{{\rm trace}}%

\global\long\def\boldone{\mathbf{1}}%

\global\long\def\hmu{\hat{\mu}}%

\global\long\def\hA{\hat{A}}%

\global\long\def\tA{\tilde{A}}%

\global\long\def\hq{\hat{q}}%

\global\long\def\Fcal{\mathcal{F}}%

\global\long\def\tmu{\tilde{\mu}}%

\global\long\def\tX{\tilde{X}}%

\global\long\def\bX{\bar{X}}%

\global\long\def\htau{\hat{\tau}}%

\global\long\def\eig{{\rm eig}}%

\global\long\def\tq{\tilde{q}}%

\global\long\def\ttau{\tilde{\tau}}%

\global\long\def\halpha{\hat{\alpha}}%

\global\long\def\hrho{\hat{\rho}}%

\global\long\def\supp{{\rm supp}}%

\global\long\def\NN{\mathbb{N}}%

\global\long\def\tf{\tilde{f}}%

\global\long\def\tC{\tilde{C}}%

\global\long\def\txi{\tilde{\xi}}%

\global\long\def\tr{\tilde{r}}%

\global\long\def\hzeta{\hat{\zeta}}%

\global\long\def\barm{\bar{m}}%

\global\long\def\texti{\text{(i)}}%

\global\long\def\textii{\text{(ii)}}%

\global\long\def\textiii{\text{(iii)}}%

\global\long\def\oneb{\mathbf{1}}%

\global\long\def\Gcal{\mathcal{G}}%

\global\long\def\Zcal{\mathcal{Z}}%

\global\long\def\talpha{\tilde{\alpha}}%

\global\long\def\tp{\tilde{p}}%

\global\long\def\Wcal{\mathcal{W}}%

\global\long\def\ZZ{\mathbb{Z}}%

\global\long\def\BB{\mathbb{B}}%

\title{New possibilities in identification of binary choice models with fixed
effects}
\author{Yinchu Zhu\thanks{Email: yinchuzhu@brandeis.edu. I am extremely grateful to Whitney
Newey for introducing me to the literature on fixed effects in nonlinear
models and for having many inspiring discussions with me. I thank
Roy Allen, Whitney Newey and Martin Mugnier for pointing out mistakes
in previous versions and for providing helpful comments on the paper.
I am also grateful for comments from Xavier D'Haultf{\oe}uille, Shakeeb
Khan, Elie Tamer and participants of various seminars and conferences.
All the errors are my own. }\\
\\
Department of Economics, \\
Brandeis University}
\maketitle
\begin{abstract}
We study the identification of binary choice models with fixed effects.
We propose a condition called sign saturation and show that this condition
is sufficient for identifying the model. In particular, this condition
can guarantee identification even when all the regressors are bounded,
including multiple discrete regressors. We also establish that without
this condition, the model is not identified unless the error distribution
belongs to a special class. Moreover, we show that sign saturation
is also essential for identifying the sign of treatment effects. Finally,
we introduce a measure for sign saturation and develop tools for its
estimation and inference. 
\end{abstract}
Key words: identification, panel model, binary choice, fixed effects

\section{Introduction}

This paper considers panel models with binary outcomes in the presence
of fixed effects. These models are convenient in economic analysis
as they allow for fairly general unobserved individual heterogeneity.
The nonlinear nature of the binary choice models makes it difficult
to eliminate the fixed effects by differencing. Since we view the
fixed effects or their conditional distribution as a nuisance parameter,
the identification of the parameter of interest becomes tricky. In
this paper, we provide new results and insights on what drives the
identification of the coefficients on the regressors and what this
means for applied work. 

Consider independent and identically distributed (i.i.d) observations
$\{(Y_{i},X_{i})\}_{i=1}^{n}$ with $Y_{i}=(Y_{i,0},Y_{i,1})$ and
$X_{i}=(X_{i,0},X_{i,1})$ from the following model
\[
Y_{i,t}=\oneb\{X_{i,t}'\beta+\alpha_{i}\geq u_{i,t}\}\qquad t\in\{0,1\},
\]
where the fixed effects are represented by the scalar variable $\alpha_{i}$
and $\beta$ is a non-random vector of coefficients. This is a semiparametric
model as the distribution of $\alpha_{i}$ given $X_{i}$ is unrestricted.
For notational simplicity, we drop the $i$ subscript in the rest
of the paper. Therefore, we write $Y=(Y_{0},Y_{1})$ and $X=(X_{0},X_{1})\in\mathcal{X}_{0}\times\mathcal{X}_{1}$
with
\begin{equation}
Y_{t}=\oneb\{X_{t}'\beta+\alpha\geq u_{t}\}\qquad t\in\{0,1\}.\label{eq: FE model}
\end{equation}

In this paper, we mainly focus on the identification of $\beta$ but
we will also discuss quantities related to treatment effects. There
are roughly two approaches to identifying $\beta$, depending on whether
or not we impose a parametric model on the distribution of $u_{t}$
given $(X,\alpha)$. Methods without parametric assumptions on the
error distribution are typically based on \citet{manski1987semiparametric}
and maximum-score-type estimation. This approach only assumes that
the distribution $u_{t}\mid(X,\alpha)$ does not depend on $t$. In
contrast, the more parametric approach relies on additional assumptions
on the functional form of the distribution $u_{t}\mid(X,\alpha)$.
For example, perhaps the most popular parametric assumption is that
$u_{t}$'s are i.i.d logistic errors across $t$ and are independent
of $(X,\alpha)$. This approach typically adopts an estimation scheme
based on the (conditional) likelihood. 

Despite the strong restrictions on the functional form, the approach
relying on parametric restrictions has received considerable attention
arguably due to identification reasons. The literature has pointed
out the widespread identification failure, such as \citet{arellano2011nonlinear}.
In particular, identification seems infeasible outside the logistic
case unless the support of $X$ is unbounded. For example, Assumption
2 in \citet{manski1987semiparametric} requires the unboundedness
of at least one component of $W=X_{1}-X_{0}$: for $W=(W_{1},...,W_{K})'$
and $\beta=(\beta_{1},...,\beta_{K})'$, there exists $k$ such that
$\beta_{k}\neq0$ and the conditional distribution $W_{k}\mid(W_{1},...,W_{k-1},W_{k+1},...,W_{K})$
has support equal to $\RR$ almost surely. Theorem 1 of \citet{chamberlain2010binary}
goes even further: for bounded $X$, the identification fails in certain
regions of the parameter space if no further restrictions are imposed
on the distribution of the error term $u_{t}$. %
\begin{comment}
The literature has pointed out the widespread identification failure,
such as \citet{arellano2011nonlinear}.
\end{comment}

In this paper, we provide a more precise picture than \citet{chamberlain2010binary}
by showing that identification without parametric assumptions on the
error distribution is quite possible even for bounded $X$. The main
motivation of this paper is to explore new possibilities without functional-form
assumptions. This raises concerns about stability. If identification
crucially hinges on the imposed functional form, the reliability of
the analysis could be a concern. After all, if the model parameter
is unidentified under every distribution function other than the logistic
one, how much should we trust this model? Hence, it is helpful to
explore a more robust setting. In this direction, this paper considers
the identification issue and leaves the estimation problem to future
research. We make the following contributions. 

First, we provide tight and simple identification conditions without
parametric restrictions on the error distribution. We show that a
sufficient condition for identification is what we refer to as sign
saturation (Assumption \ref{assu: sign saturation}), which states
that $P\left(E(Y_{1}-Y_{0}\mid X)>0\right)$ and $P\left(E(Y_{1}-Y_{0}\mid X)<0\right)$
are both strictly positive. This condition can hold with bounded regressors
$X$ and allow for multiple discrete regressors and interaction terms.
Moreover, we show that this is also a necessary condition for identification
unless the distribution of $u_{t}$ is in a special class. Therefore,
although we know (from \citet{chamberlain2010binary}) that for bounded
regressors the identification fails at some values of $\beta$, our
results pinpoint these values: the identification fails exactly at
points where the sign saturation fails. Therefore, whether the regressors
are bounded is not what really drives the identification. The identifiability
is more closely related to the sign saturation condition, which is
simple and intuitive. 

A key advantage of the proposed sign saturation condition is that
it is stated in terms of observed variables only and does not involve
the parameter that we are trying to identify. Hence, it can be used
to answer the question of whether identification holds under the distribution
of observed data, thereby making identification directly testable.
This is different from the \citet{manski1987semiparametric}-type
condition. For example, suppose that $K=2$, $W_{1}$ is bounded and
the conditional distribution $W_{2}\mid W_{1}$ has support $\RR$
almost surely. The identification condition in \citet{manski1987semiparametric}
becomes $\beta_{2}\neq0$. However, how can we check $\beta_{2}\neq0$
when the identifiability of $\beta$ is not yet established? It does
not seem obvious how to do so with the data. In contrast, under the
sign saturation condition, we check statements \textit{only} on observed
data: $P\left(E(Y_{1}-Y_{0}\mid X)>0\right)$ and $P\left(E(Y_{1}-Y_{0}\mid X)<0\right)$. 

Second, we show that the sign saturation condition is also sufficient
and necessary for learning the sign of the marginal effects. Even
though the magnitude of different notions of the marginal effects
is usually not point identified, the sign of the effects is typically
the same as the sign of a component of $\beta$. It turns out that
outside a special class of distributions, if the sign saturation condition
fails, we cannot guarantee to distinguish between zero effects and
strictly positive effects. In many empirical studies, one important
task is to check whether the identified set for the effects includes
zero. Therefore, it is worthwhile to check the sign saturation condition
before exploring additional assumptions that give bounds to treatment
effects. 

Third, we introduce a tool that can be used to assess the sign saturation
condition empirically. We show that $E\max\{P(Y_{1}-Y_{0}\mid X),0\}$
and $E\min\{P(Y_{1}-Y_{0}\mid X),0\}$ can be directly estimated from
the data and confidence intervals can be constructed by a simple bootstrap
without any tuning parameters such as bandwidth or number of basis
functions. For bounded regressors, the existing literature highlights
the risk of identification failure: the identification possibly (not
definitely) fails. This generic warning may have discouraged the application
of fixed effect models in practice. The proposed tool serves as a
more accurate diagnosis so that the identifiability can be assessed
from the data.%
\begin{comment}
It turns out that the computation is already familiar to applied researchers
in this area as it only involves computing the usual maximum score
estimator and its bootstrapping. We are not aware of any existing
tools for checking identification in this model.
\end{comment}
\begin{comment}
One important objective of our work is to provide robust tools for
empirical research. Although conditional maximum likelihood or any
other approach based on logistic errors is popular in applied work,
we argue that this assumption can be problematic in applications.
If so, it might be prudent to leave the error distribution nonparametric.
Since estimation and inference based on maximum score estimates are
already available (e.g., \citet{kim1990cube}, \citet{seo2018local}
and \citet{cattaneo2020bootstrap}), we only need to make sure that
we can reasonably assume the identification. This paper fills this
very need: tight and testable conditions for identification. 
\end{comment}

Our work is closely related to the literature of binary response models
with logistic errors. \citet{rasch1960probabilistic}, \citet{andersen1970asymptotic}
and \citet{chamberlain1980analysis} consider the estimation and inference
in this case. The logistic link function is also singled out in \citet{chamberlain2010binary}
as the ``nice'' link function: the identification does not require
unbounded regressors. Recently, the work of \citet{Mugnier2009.08108}
points out that if there are more than 2 time periods, then the nice
link function can be extended to a generalized version of the logistic
distribution. This suggests that imposing a special class of parametric
structures is a useful strategy of achieving identification and this
special class seems to be the logistic functions. We provide new insight
on this. We show that this special class also includes those such
that $\dot{G}(\cdot)$ is periodic, where $\dot{G}(\cdot)$ is the
derivative of $\ln\frac{F(\cdot)}{1-F(\cdot)}$ and $F(\cdot)$ is
the distribution function of $u_{t}$. If $F(\cdot)$ is logistic,
then $\dot{G}(\cdot)$ is a constant function, which is clearly periodic.
However, we find that identification is possible for any periodic
$\dot{G}(\cdot)$. Although \citet{chamberlain2010binary} tells us
that for non-logistic distributions, identification fails in an open
neighborhood, we show that for any periodic $\dot{G}(\cdot)$, identification
must hold in another open neighborhood. Therefore, it seems to us
that the truly problematic distributions are those with non-periodic
$\dot{G}(\cdot)$. Indeed, outside this special class of periodic
$\dot{G}(\cdot)$, we show that identification is impossible at \textit{every}
point if the sign saturation condition is not satisfied. Thus, for
generic distribution functions, sign saturation is a necessary condition
for identification. We summarize the relation between identifiability
and error distribution in Table \ref{tab: ID and sign sat}. %
\begin{comment}
Logistic errors are also assumed in the study of identification in
dynamic models, including recent works of \citet{Honore2005.05942},
\citet{khan2020identification} among many others.
\end{comment}

Our work is also closely related to works that do not assume logistic
errors. First, some results in the literature allow for bounded regressors
but require all regressors to be continuous. For example, Assumption
3.3 of \citet{shi2018estimating} gives a sufficient condition for
identification under bounded support. As commented in the paper, their
assumption essentially requires all regressors to be continuous. Continuous
distribution on all the regressors are also required by Assumption
6 of \citet{toth2017} and Assumption 4' of \citet{gao2023logical}.
Examples on nonseparable models include \citet{hoderlein2012nonparametric},
\citet{chernozhukov2015nonparametric} and \citet{chernozhukov2019nonseparable};
their arguments rely on the derivatives with respect to $X$ and thus
$X$ needs to be continuous. Here, we allow for one or multiple discrete
regressors, which are important in empirical studies as many treatment
variables are binary. Second, Corollary 4.1 of \citet{horowitz2009semiparametric}
considers a related model in the cross-sectional setting. It is possible
to translate this to the panel data in terms of conditional densities,
but these are still not the most general conditions; for example,
some strictly increasing continuous distribution functions have zero
density almost everywhere, e.g., \citet{salem1943some} and \citet{takacs1978increasing}.
Most importantly, none of the aforementioned works show whether their
conditions are necessary. We contribute to the literature by finding
what must be assumed and developing empirical tools to assess it.%
{} %
\begin{comment}
Notice that this condition rules out a discrete $X_{1}$; if $X_{1}$
is discrete, then the density of $X_{1}+X_{-1}'\beta_{-1}$ conditional
on $X_{-1}=x_{-1}$ does not exist. This can be viewed as a \textbf{conditional}
sign saturation, namely, conditional on $X_{-1}$, $X'\beta$ has
sign saturation. Our sign saturation is an \textbf{unconditional}
statement. The conditional sign saturation is stronger than necessary
even for continuous $X$; in Lemma \ref{lem: example horowitz} in
the appendix, we give an example with $X=(X_{1},X_{2})'$ in which
almost surely, one of $P(X'\beta>0\mid X_{1})$ and $P(X'\beta<0\mid X_{1})$
is exactly zero and one of $P(X'\beta>0\mid X_{2})$ and $P(X'\beta<0\mid X_{2})$
is exactly zero.
\end{comment}
{} %

Our necessity results complement those in \citet{pakes2024moment}.
By Proposition 3 therein, $\arg\max_{\beta}E(Y_{1}-Y_{0})\cdot\oneb\{W'\beta\geq0\}$,
the set of maximizers of the maximum score criterion function, is
the sharp identification set when we only assume the stationarity
of the distribution $u_{t}\mid(X,\alpha)$. Our results show that
this sharp identified reduces to a singleton up to scale under the
sign saturation condition. A natural question is the necessity of
this condition, especially if we are willing to impose additional
assumptions, namely i.i.d errors that are independent of $(X,\alpha)$.
By classical results, we know that under logistic errors, $\arg\max_{\beta}E(Y_{1}-Y_{0})\cdot\oneb\{W'\beta\geq0\}$
is not the sharp identified set. Our results in Section \ref{subsec: necessary cond}
further complete the puzzle by characterizing point identification
as the periodicity of $\dot{G}(\cdot)$. %
\begin{comment}
However, the proof of their Theorem 2 shows that every point in the
identified set can be rationalized by a stationary error $u_{t}\mid X$
without any fixed effects $P(\alpha=0)=1$. Therefore, the class of
models with arbitrary stationary errors is so large that the fixed
effects are not important any more. In this sense, our sign saturation
condition is necessary for point identification up to scale. \footnote{One can show that the maximum score criterion function is maximized
at a unique point (up to scale) under the sign saturation condition
if the support of $Z$ is convex with non-empty interior.} 
\end{comment}
{} 

\section{\label{sec: id result}Identification of coefficients}

In the setting of (\ref{eq: FE model}), \citet{manski1987semiparametric}
identifies $\beta$ with the distributional stationarity condition:
\begin{assumption}
\label{assu: stationary error}The distribution of $u_{t}\mid(X,\alpha)$
does not depend on $t\in\{0,1\}$. Let the cumulative distribution
function (c.d.f) of this common distribution be $F(\cdot|x,\alpha)$.
Assume that for any $(x,\alpha)$, $F(\cdot\mid x,\alpha)$ is a strictly
increasing function on $\RR$.
\end{assumption}
Assumption \ref{assu: stationary error} rules out dynamic models.
In this paper, we focus on static models with two time periods. Without
any restriction on the magnitude of $\alpha$ and $u_{t}$ in (\ref{eq: FE model}),
it is impossible to identify the magnitude of $\beta$ but the following
result by \citet{manski1987semiparametric} gives the identification
of the ``direction'' of $\beta$; in other words, $\beta$ is identified
up to scaling. 
\begin{lem}[\citet{manski1987semiparametric}]
\label{lem: manski lem 1}Let Assumption \ref{assu: stationary error}
hold. Then with probability one, 
\begin{equation}
{\rm sgn}\left(E(Y_{1}-Y_{0}\mid X)\right)={\rm sgn}\left(W'\beta\right),\label{eq: manski condition}
\end{equation}
where $W=X_{1}-X_{0}$ and ${\rm sgn}(\cdot)$ is defined as ${\rm sgn}(t)=1$
for $t>0$, ${\rm sgn}(t)=-1$ for $t<0$ and ${\rm sgn}(0)=0$.
\end{lem}
The key idea of our result is based on the condition in (\ref{eq: manski condition}).
Clearly, the conditional mean function $E(Y_{1}-Y_{0}\mid X)$ is
identified. For simplicity, suppose that $P(E(Y_{1}-Y_{0}\mid X)=0)=0$.
Then (\ref{eq: manski condition}) implies that there is a hyperplane
of $W$ that perfectly classifies ${\rm sgn}(E(Y_{1}-Y_{0}\mid X))$.
This is an ideal support vector machine (SVM) as illustrated in Figure
\ref{fig:SVM}, where the red and blue dots represent $-1$ and $1$,
respectively for ${\rm sgn}(E(Y_{1}-Y_{0}\mid X))$.\footnote{In \citet{komarova2013binary}, the geometry based on SVM is also
to study binary response models with a median restriction; the model
there does not have a panel structure. We thank Christopher Walker
for this reference. } If we want to identify the classification boundary in the SVM, we
obviously need to have both classes (red and blue) in the data. Since
these two classes correspond to $E(Y_{1}-Y_{0}\mid X)>0$ and $E(Y_{1}-Y_{0}\mid X)<0$,
this simple requirement is formalized as the following sign saturation
condition.

\begin{figure}
\caption{\label{fig:SVM}Lemma \ref{lem: manski lem 1} as a support vector
machine}

\begin{centering}
\includegraphics[clip,scale=0.6]{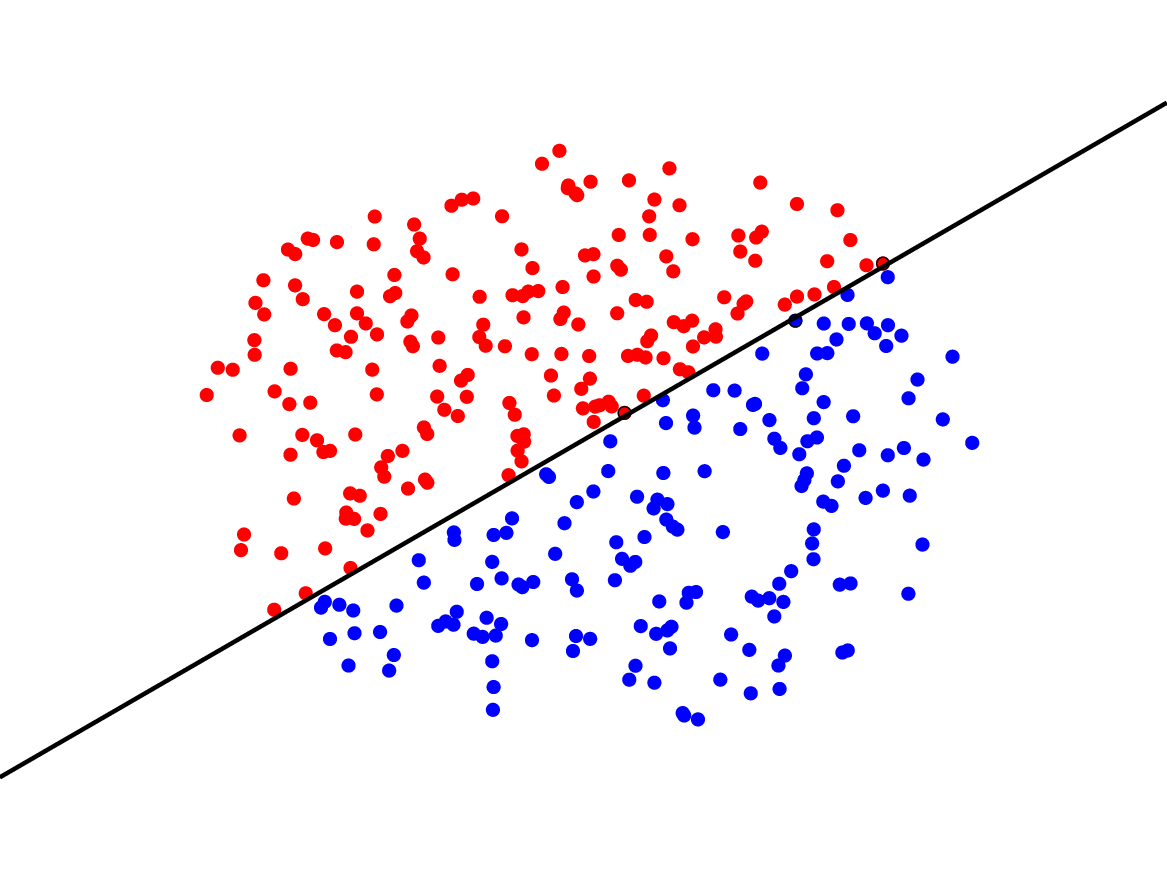}
\par\end{centering}
{\small{}This demonstrates the case with $W\in\RR^{2}$. The red dots
correspond to ${\rm sgn}(W'\beta)=1$ and the blue ones correspond
to ${\rm sgn}(W'\beta)=-1$.}{\small\par}
\end{figure}

\begin{assumption}
\label{assu: sign saturation}Both $P(E(Y_{1}-Y_{0}\mid X)>0)$ and
$P(E(Y_{1}-Y_{0}\mid X)<0)$ are strictly positive. 
\end{assumption}
From Figure \ref{fig:SVM}, it is also clear that in addition to having
both classes, we also need the data points to be ``dense'' around
the hyperplane. For simplicity, we will require that the support of
$W$ is convex and has no-empty interior so the support of $W$ is
dense. This is all we need for uniquely determining the hyperplane
and there is nothing about unbounded supports. We will formalize this
intuition in in the next subsection and extend the results to the
case of multiple discrete regressors in Section \ref{sec: multiple discrete}.
\begin{rem}
Phrasing the problem in terms of support vector machines makes it
easy to see that when all the regressors are discrete, we cannot in
general expect to point identify $\beta$ up to scale. To see this,
notice that by \citet{pakes2024moment}, the sharp identified set
for $\beta$ is $\arg\max_{b}E(Y_{1}-Y_{0}){\rm sgn}(W'b)$, which
can be written as 
\[
\arg\max_{b}E\left[\left|E(Y_{1}-Y_{0}\mid X)\right|\cdot{\rm sgn}(W'\beta)\cdot{\rm sgn}(W'b)\right].
\]
Therefore, if ${\rm sgn}(W'\tilde{\beta})={\rm sgn}(W'\beta)$ with
probability one, then $\tilde{\beta}$ is in the identified set. When
the regressors are discrete, we can in general move $\beta$ a bit
and retain the same ${\rm sgn}(W'\beta)$. As illustrated in Figure
\ref{fig:SVM 2}, the black and gray lines both perfectly classify
the red and blue dots and thus represent two points in the identified
set that are not scalar multiple of each other.
\end{rem}
\begin{figure}
\caption{\label{fig:SVM 2}Discrete regressors as a support vector machine }

\begin{centering}
\includegraphics[clip,scale=0.6]{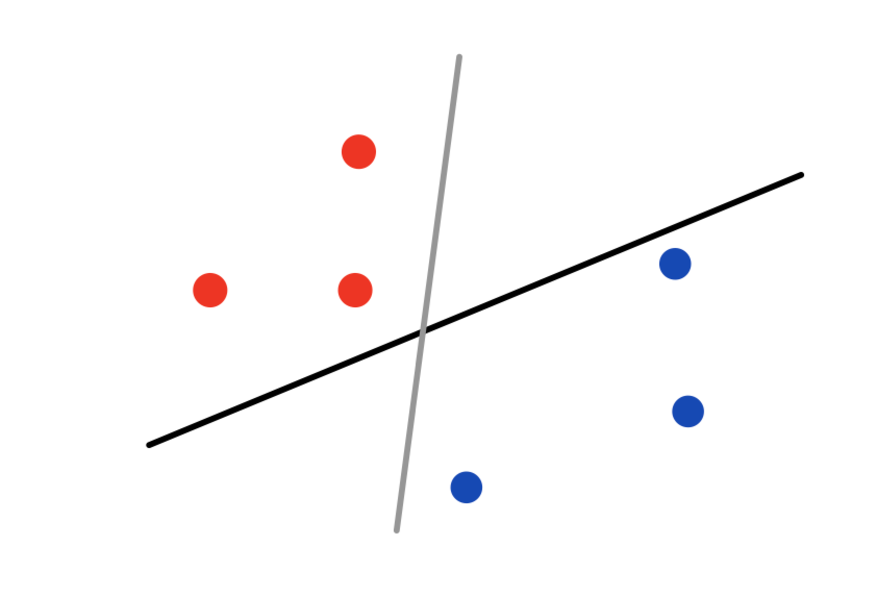}
\par\end{centering}
{\small{}This demonstrates the case in which $W\in\RR^{2}$ is discrete.
The red dots correspond to ${\rm sgn}(W'\beta)=1$ and the blue ones
correspond to ${\rm sgn}(W'\beta)=-1$.}{\small\par}
\end{figure}

\subsection{The sign saturation condition is sufficient for identification}

Following \citet{chamberlain2010binary}, we assume that one of the
regressors is binary in that it is equal to zero at $t=0$ and is
equal to one at $t=1$; in other words, one component of $W=X_{1}-X_{0}$
is one. Thus, without loss of generality, we can partition $W=(Z',1)'\in\RR^{K}$.
If all the regressors are continuous, then the result can still be
applied because the coefficient for the binary regressor is allowed
to be zero, see Corollary \ref{cor: only cont X}. The first main
result of this paper is to show that the sign saturation condition
(together with additional weak assumptions) is enough to guarantee
identification up to scaling. To state the formal result, we first
recall the definition of the support of a random variable (or its
corresponding probability measure): the support of probability measure
$\lambda$ is the smallest closed set $A$ such that $\lambda(A)=1$,
e.g., page 227 of \citet{dudley_2002}.
\begin{thm}
\label{thm: mixed reg}Let Assumption \ref{assu: stationary error}
hold. Partition $W=(Z',1)'\in\RR^{K}$. Suppose that the support of
$Z$ is convex and has non-empty interior. Let Assumption \ref{assu: sign saturation}
hold. 

Then $b=\mu\beta$ for some $\mu>0$ if and only if $R(b)=0$, where
\begin{equation}
R(b)=P\left({\rm sgn}\left(E(Y_{1}-Y_{0}\mid X)\right)\neq{\rm sgn}(W'b)\right).\label{eq: R function}
\end{equation}

Therefore, $\beta$ is identified up to scale. 
\end{thm}
The requirement on the distribution of $Z$ is mild. Let $\Zcal$
be the support of $Z$. First, we do not require $\Zcal$ to be an
unbounded set. Second, the distribution of $Z$ does not have to admit
a density and ``atoms'' (or point masses) are allowed. Theorem \ref{thm: mixed reg}
imposes this via the convexity and non-empty interior of $\Zcal$.
When $\Zcal$ is not convex, the result is still useful: as long as
$\Zcal$ contains a convex subset with non-empty interior, we can
apply the result on this subset by restricting the sample, i.e., the
sign saturation condition holds on the restricted sample. If there
are multiple discrete variables, then we set $Z$ to be the difference
in continuous variables, see Section \ref{sec: multiple discrete}.

By Theorem \ref{thm: mixed reg}, to check whether $b$ is a rescaled
version of the true $\beta$, we only need to check $R(b)$ from the
distribution of the observed data. Notice that the identification
does not assume that $u_{t}$'s are independent across $t$; this
is because Assumption \ref{assu: stationary error} only requires
them to have the same marginal distribution given $(X,\alpha)$. To
give some intuition on the identifying power of sign saturation under
bounded $\Zcal$, we use the following simple example to illustrate
why the criterion function $R(\cdot)$ fails to deliver identification
when Assumption \ref{assu: sign saturation} is not satisfied. 
\begin{example}
\label{exa: time dummy}Suppose that $W=(Z,1)$, where $Z$ is a continuous
random variable with support $[-1,1]$. For simplicity, we consider
$\beta=(\beta_{1},\beta_{2})'$ and $b=(b_{1},b_{2})'$ with $\beta_{1},b_{1}>0$.
The question is whether we can derive $b=\mu\beta$ for some $\mu>0$
from $R(b)=0$. Here, identifying $\beta$ up to scale is equivalent
to identifying $\beta_{2}/\beta_{1}$. After straight-forward calculations,
we can see that 
\[
R(b)=P\left(Z\in\left(-\frac{\beta_{2}}{\beta_{1}},-\frac{b_{2}}{b_{1}}\right)\bigcup\left(-\frac{b_{2}}{b_{1}},-\frac{\beta_{2}}{\beta_{1}}\right)\right).
\]

Suppose that the sign saturation condition fails, say $P(W'\beta>0)=0$,
which implies that $-\beta_{2}/\beta_{1}\geq1$. In this case, the
set $\left(-\frac{\beta_{2}}{\beta_{1}},-\frac{b_{2}}{b_{1}}\right)\bigcup\left(-\frac{b_{2}}{b_{1}},-\frac{\beta_{2}}{\beta_{1}}\right)$
has empty intersection with $[-1,1]$ whenever $-b_{2}/b_{1}>1$.
In other words, if $P(W'\beta>0)=0$, then $R(b)=0$ whenever $-b_{2}/b_{1}>1$;
as a result, we cannot always verify $b_{2}/b_{1}=\beta_{2}/\beta_{1}$
through $R(b)$. 

On the other hand, if the sign saturation condition holds, then we
have $-1<\beta_{2}/\beta_{1}<1$. In this case, the intersection of
$\left(-\frac{\beta_{2}}{\beta_{1}},-\frac{b_{2}}{b_{1}}\right)\bigcup\left(-\frac{b_{2}}{b_{1}},-\frac{\beta_{2}}{\beta_{1}}\right)$
and $[-1,1]$ always has non-empty interior whenever $b_{2}/b_{1}\neq\beta_{2}/\beta_{1}$.
Consequently, $R(b)=0$ if and only $b_{2}/b_{1}=\beta_{2}/\beta_{1}$.
\qed
\end{example}
The following example demonstrates how Theorem 1 gives identification
conditions for $\beta$ with interaction terms. 
\begin{example}[Interaction terms]
Let $D_{t}=\oneb\{t=1\}$ and consider a continuous variable $H_{t}$.
Suppose that we are interested in the treatment effect of $D_{t}$
controlling for $H_{t}$. For a more flexible specification, we include
the interaction $D_{t}H_{t}$. Then $X_{t}=(D_{t},H_{t},D_{t}H_{t})'\in\RR^{3}$
with the corresponding partition $\beta=(\beta_{1},\beta_{2},\beta_{3})'\in\RR^{3}$.
For simplicity, suppose that the support for $(H_{0},H_{1})'$ is
$[-1,1]\times[-1,1]$. Then $Z=(H_{1},D_{1}H_{1})'-(H_{0},D_{0}H_{0})'=(H_{1}-H_{0},H_{1})'$.
One can easily verify that $\Zcal$ is convex and has non-empty interior.\footnote{To see the convexity, simply notice that $Z=\begin{pmatrix}-1 & 1\\
0 & 1
\end{pmatrix}\begin{pmatrix}H_{0}\\
H_{1}
\end{pmatrix}$ and $(H_{0},H_{1})'$ has a convex support.} By Lemma \ref{lem: manski lem 1}, straight-forward calculations
show that the sign saturation condition becomes $|\beta_{2}+\beta_{3}|+|\beta_{3}|>|\beta_{1}|$.
For example, this condition holds when $\beta_{1}=\beta_{3}=0$ and
$\beta_{2}\neq0$. Therefore, to test the null hypothesis of ``zero
effects whatsoever'', it suffices to have a relevant control variable.
\qed
\end{example}
We now discuss an important special case with only continuous covariates.
\begin{cor}
\label{cor: only cont X}Let Assumption \ref{assu: stationary error}
hold. If the interior of the support of $X_{1}-X_{0}$ contains $\boldsymbol{0}_{\dim(X_{t})}$
(zero in $\RR^{\dim(X_{t})})$ and $\beta\neq\boldsymbol{0}_{\dim(X_{t})}$,
then $\beta$ is identified up to scale. 
\end{cor}
A comparison between Example \ref{exa: time dummy} and Corollary
\ref{cor: only cont X} highlights the role of discrete variables
in identification. If all the covariates are continuous and the change
in covariates contains zero in the support, then we can generally
expect identification (up to scale) of $\beta$, regardless of whether
the covariates are bounded. In contrast, if the covariates include
discrete variables (such as the time fixed effect in \citet{chamberlain2010binary}),
we no longer have such a generic guarantee of identification. 

\subsection{\textcolor{blue}{\label{sec: multiple discrete}}Extension to multiple
discrete regressors}

In this subsection, we consider the case of multiple discrete regressors.\footnote{I thank Xavier D'Haultf{\oe}uille for a discussion of this case.}
We partition $X_{t}=(X_{(1),t}',X_{(2),t}')'$, where $X_{(1),t}\in\RR^{K_{1}}$
are discrete regressors and $X_{(2),t}\in\RR^{K_{2}}$ are continuous
regressors. We can partition the corresponding $W=X_{1}-X_{0}=(D',Z')'$
with $D=X_{(1),1}-X_{(1),0}$ and $Z=X_{(2),1}-X_{(2),0}$. Let $\Dcal$
denote the support of $D$. We show that identification still holds
under a conditional version of sign saturation. 
\begin{thm}
\label{thm: discrete cov ID}Let Assumption \ref{assu: stationary error}
hold. Suppose that there exist linearly independent $d_{1},...,d_{K_{1}}\in\Dcal$
with such that for $j\in\{1,...,K_{1}\}$,\\
(1) both $P(E(Y_{1}-Y_{0}\mid X,D=d_{j})>0)$ and $P(E(Y_{1}-Y_{0}\mid X,D=d_{j})<0)$
are strictly positive\\
(2) the support of $Z$ conditional on $D=d_{j}$ is convex and has
non-empty interior.

Then $\beta$ is identified up to scale.
\end{thm}
We note that $|\Dcal|$ can be much larger than $K_{1}$. For example,
if $X_{(1),t}$ represents $K_{1}$ binary variables, then the support
of $X_{(1),t}$ is $\{0,1\}^{K_{1}}$ and $\Dcal=\{-1,0,1\}^{K_{1}}$,
which means $|\Dcal|=3^{K_{1}}$. Theorem \ref{thm: discrete cov ID}
is convenient in that it does not require the sign saturation to hold
conditional on $D=d$ for \textbf{every} $d\in\Dcal$. It suffices
to require this for $K_{1}$ linearly independent elements of $\Dcal$.
If $K_{1}=1$ ($X_{(1),t}$ is a time dummy), then this requirement
reduces to Assumption \ref{assu: sign saturation}. Theorem \ref{thm: discrete cov ID}
also allows for more complicated cases as demonstrated in the following
example.
\begin{example}[Time dummy and categorical variables]
Suppose that $R_{t}$ is a categorical variable taking values in
$\{1,...,K_{1}\}$. Consider $X_{(1),t}=(\oneb\{t=1\},\oneb\{R_{t}=1\},...,\oneb\{R_{t}=K_{1}-1\})'\in\RR^{K_{1}}$,
i.e., a time dummy and $K_{1}-1$ dummy variables for $R_{t}$. Then
we can write the support of $X_{(1),t}\in\RR^{K_{1}}$ as $\{\oneb\{t=1\}\}\times\{e_{1},...,e_{K_{1}-1},\boldsymbol{0}_{K_{1}-1}\}$,
where $e_{j}$ is the $j$-th column of the $(K_{1}-1)\times(K_{1}-1)$
identity matrix and $\boldsymbol{0}_{K_{1}-1}=(0,...,0)'\in\RR^{K_{1}-1}$.
Let us assume that the support of $(R_{0},R_{1})'$ contains $(K_{1},j)$
for any $j\in\{1,...,K_{1}\}$. Then $\Dcal$ contains $\{1\}\times\{e_{1},...,e_{K_{1}-1},\boldsymbol{0}_{K_{1}-1}\}$.
In other words, $\Dcal$ contains all the columns of the following
matrix: 
\[
\begin{pmatrix}1 & 1 & 1 & \cdots & 1\\
0 & 1 & 0 & \cdots & 0\\
0 & 0 & 1 & \cdots & 0\\
\vdots & \vdots & \vdots & \ddots & \vdots\\
0 & 0 & 0 & \cdots & 1
\end{pmatrix}\in\RR^{K_{1}\times K_{1}}.
\]

We can easily see that the determinant of this upper triangular matrix
is equal to one. Therefore, the above matrix has full rank, which
means that $\Dcal$ contains $K_{1}$ linearly independent elements.
Then $\beta$ is identified up to scale if for every $j\in\{1,...,K_{1}\}$,
$P(E(Y_{1}-Y_{0}\mid Z,R_{0}=K_{1},R_{1}=j)>0)$ and $P(E(Y_{1}-Y_{0}\mid Z,R_{0}=K_{1},R_{1}=j)<0)$
are strictly positive and the support of $Z$ conditional on $(R_{0},R_{1})=(K_{1},j)$
is convex and has non-empty interior.\qed
\end{example}

\subsection{\label{subsec: necessary cond}Is the sign saturation condition is
also necessary?}

It turns out that in the many situations, sign saturation characterizes
identifiability. To illustrate this point, we consider the setting
studied by \citet{chamberlain2010binary} and assume that $u_{t}$
is independent of $(X,\alpha)$ and is from a known distribution.\footnote{To establish necessity, it is enough to consider the case in which
$u_{t}$ is independent of $(X,\alpha)$. This is similar to showing
minimax lower bounds. If identification cannot be guaranteed even
with the additional assumption of independence between $u$ and $(X,\alpha)$,
then it is definitely not guaranteed without this assumption.} We show that without sign saturation, identification fails at every
point unless the distribution of $u_{t}$ is from a special class. 

Suppose that in time period $t\in\{0,1\}$, we observe $(Y_{t},X_{t})$,
where $X_{t}=(X_{t,1}',X_{t,2})'$ with $X_{t,1}\in\RR^{K-1}$ and
$X_{t,2}=\oneb\{t=1\}$. Suppose that $u_{t}$ is independent of $(X,\alpha)$
and has c.d.f $F(\cdot)$. Then from the data, the distribution of
$Y$ given $(X,\alpha)$ is determined by the vector
\begin{align}
L(X;\beta,\alpha) & =\begin{pmatrix}F(X_{0,1}'\beta_{1}+\alpha)\\
F(X_{1,1}'\beta_{1}+\beta_{2}+\alpha)\\
F(X_{0,1}'\beta_{1}+\alpha)\cdot F(X_{1,1}'\beta_{1}+\beta_{2}+\alpha)
\end{pmatrix}\nonumber \\
 & =\begin{pmatrix}F(X_{0,1}'\beta_{1}+\alpha)\\
F(X_{0,1}'\beta_{1}+\alpha+W'\beta)\\
F(X_{0,1}'\beta_{1}+\alpha)\cdot F(X_{0,1}'\beta_{1}+\alpha+W'\beta)
\end{pmatrix},\label{eq: L fun}
\end{align}
where $\beta=(\beta_{1}',\beta_{2})$ is partitioned as $\beta_{1}\in\RR^{K-1}$
and $\beta_{2}\in\RR$, and $W:=X_{1}-X_{0}=(Z',1)'$ with $Z=X_{1,1}-X_{0,1}$.
As pointed out in \citet{chamberlain2010binary}, here we should aim
for identification of $\beta$, not just up to scaling, because ``our
scale normalization is built in to the given specification for the
$u_{t}$ distribution''. Let $\Zcal$ be the support of $Z$, i.e.,
again the smallest closed set with the full probability mass. Since
$W=(Z',1)'$, the support of $W$ is $\Wcal=\Zcal\times\{1\}$.

The special class of functions is best described in terms of a transformation
of $F(\cdot)$. Define the function $G(\cdot)=\ln\frac{F(\cdot)}{1-F(\cdot)}$.
Let $\dot{G}(\cdot)$ denote the derivative of $G(\cdot)$. If $F(\cdot)$
is the logistic function, then $G(\cdot)$ is an affine function and
$\dot{G}(\cdot)$ is a constant function. It turns out that if we
are willing to rule out cases with $\dot{G}(\cdot)$ being a periodic
function,\footnote{A function $h(\cdot)$ on $\RR$ is a periodic function if there exists
$c\neq0$ such that $h(c+t)=h(t)$ for any $t\in\RR$. In this case,
$c$ is said to be a period of $h(\cdot)$. Notice that we assume
that a period has to be non-zero; otherwise, every function would
be a periodic function with a period being zero.} then sign saturation is a necessary condition for identification.
This special class includes the logistic function, whose $\dot{G}(\cdot)$
is a constant function and is clearly periodic. (Any real number is
a period of a constant function.)

We now introduce the notations for describing identification. Let
$\Pi$ denote the set of probability measures on $\RR$. The distribution
of $Y\mid X=x$ is determined by $\int L(x;\beta,\alpha)d\pi_{x}(\alpha)$
for some $\pi_{x}\in\Pi$. Here, $\pi_{x}$ denotes the distribution
$\alpha\mid X=x$. We now recall the notation of identification failure
discussed in \citet{chamberlain2010binary}.
\begin{defn}[Identification failure at $\beta$]
\label{def: chamberlain def}We say that the identification fails
at $\beta$ if there exists $b\neq\beta$ such that for any $x$ in
the support, there exist $\pi_{x},\tilde{\pi}_{x}\in\Pi$ depending
on $x$ such that $\int L(x;\beta,\alpha)d\pi_{x}(\alpha)=\int L(x;b,\alpha)d\tilde{\pi}_{x}(\alpha)$.
\end{defn}
This definition says that there exist two ``mixing distributions''
$\pi_{x}$ and $\tilde{\pi}_{x}$ in $\Gcal$ such that the two mixtures
(representing the distribution $Y\mid X=x$) $\int L(x;\beta,\alpha)d\pi_{x}(\alpha)$
and $\int L(x;b,\alpha)d\tilde{\pi}_{x}(\alpha)$ are identical. Following
\citet{chamberlain2010binary}, we can state identification in terms
of convex hulls of $L(x;\beta,\alpha)$. The identification fails
at $\beta$ if there exists $b\neq\beta$ such that ${\rm conv}\{L(x;\beta,\alpha):\alpha\in\RR\}\bigcap{\rm conv}\{L(x;b,\alpha):\alpha\in\RR\}\neq\emptyset$
for any $x$, where ${\rm conv}$ denotes the convex hull. 

Notice that $X_{0,1}'\beta_{1}$ can be absorbed by $\alpha$ since
the distribution of $\alpha$ is allowed to have arbitrary dependence
on $x$. Without loss of generality, we view the entire $X_{0,1}'\beta_{1}+\alpha$
as the fixed effects to simplify notations. Then we can restate Definition
\ref{def: chamberlain def} as follows.
\begin{defn}
\label{def: chamb alter}We say that the identification fails at $\beta$
if there exists $b\neq\beta$ such that $\Acal(w'\beta)\bigcap\Acal(w'b)\neq\emptyset$
for any $w\in\Wcal$, where for any $t\in\RR$, $\Acal(t)={\rm conv}\{p(t,\alpha):\alpha\in\RR\}$
and 
\[
p(t,\alpha):=\begin{pmatrix}F(\alpha)\\
F(t+\alpha)\\
F(\alpha)\cdot F(t+\alpha)
\end{pmatrix}.
\]
\end{defn}
Perhaps the simplest way to see the equivalence between the two definitions
is to notice that ${\rm conv}\{L(x;\beta,\alpha):\alpha\in\RR\}=\Acal(w'\beta)$,
where $w=x_{1}-x_{0}$. We now define the parameter space in which
the sign saturation fails: $\Bcal_{+}=\{\beta\in\RR^{K}:\ w'\beta>0\ \forall w\in\Wcal\}$.
The main result for the necessity is the following. 
\begin{thm}
\label{thm: neccessity part 1}Suppose that $\Zcal$ is bounded. Suppose
that $F(\cdot)$ is continuously differentiable with support on $\RR$.
If $\dot{G}(\cdot)$ is not a periodic function, then the identification
fails at every point in $\Bcal_{+}$. 
\end{thm}
We compare this with Theorem 1 of \citet{chamberlain2010binary},
which states that if $F(\cdot)$ is outside a special class, identification
fails in an open neighborhood. Theorem \ref{thm: neccessity part 1}
complements this result by showing that if $F(\cdot)$ is outside
a special class, identification fails at every point, not just in
an open neighborhood. This is also why the proof of Theorem \ref{thm: neccessity part 1}
follows a very different strategy from the proof of Theorem 1 of \citet{chamberlain2010binary};
the latter essentially only needs to find one point at which the identification
fails whereas we need to show that identification fails universally
in $\Bcal_{+}$. 

The special class in \citet{chamberlain2010binary} only includes
logistic functions, but here our special class is larger and includes
all functions with periodic $\dot{G}(\cdot)$. This enlargement of
the special class cannot be avoided because we show that there are
instances of non-logistic $F(\cdot)$ for which identification holds,
e.g., $F(t)=[1+\exp(-2t-\sin(t))]^{-1}$. We can push our arguments
further and show that the identification under periodic $\dot{G}(\cdot)$
is ``robust'' and is based on a set with strictly positive probability
mass. 
\begin{thm}
\label{thm: ID period fun}Suppose that $\dot{G}(\cdot)$ is a continuous
periodic function that is non-constant. Let $\beta=(\beta_{1}',\beta_{2})'\in\Bcal_{+}$.
Then\\
(1) there exists a minimal positive period $\eta>0$ for $\dot{G}(\cdot)$,
i.e., the set $\{a>0:\ \dot{G}(a+t)=\dot{G}(t)\ \forall t\in\RR\}$
has a smallest element. \\
(2) If $(z'\beta_{1}+\beta_{2})/\eta$ is an integer for some $z$
in the interior of $\Zcal$, then $\beta$ is identified; moreover,
for any $b\in\Bcal_{+}$ with $b\neq\beta$, $P\left(\Acal(W'\beta)\bigcap\Acal(W'b)=\emptyset\right)>0$,
where $W=(Z',1)'$. \\
(3) If $\Zcal$ is compact and has non-empty interior, then there
exists an open ball $D\subset\Bcal_{+}$ such that every point in
$D$ is identified. 
\end{thm}
By Theorem \ref{thm: ID period fun}, as long as $\dot{G}(\cdot)$
is periodic and $z'\beta_{1}+\beta_{2}$ is in $\eta\cdot\NN$ for
some interior point $z$ ($\NN$ denotes the set of positive integers),
we still have identification. This complements Theorem 1 of \citet{chamberlain2010binary}
in an interesting way. For non-logistic $F(\cdot)$, it is true that
the identification fails at every point in an open set. On the other
hand, we show that the identification also holds at every point in
another open set for periodic $\dot{G}(\cdot)$. It is worth noting
that in Theorem \ref{thm: ID period fun}, the identification of $\beta$
is robust in the sense that it is not based on a small number of ``unimportant
points'' in $\Wcal$; the points in $\Wcal$ that allow us to identify
$\beta$ have strictly positive probability mass. Therefore, the truly
hopeless cases for identification are those with non-periodic $\dot{G}(\cdot)$.
Finally, for the case of compact and convex $\Zcal$ with non-empty
interior, we summarize our results and the existing literature in
Table \ref{tab: ID and sign sat}. 

\begin{table}[h]
\caption{\label{tab: ID and sign sat}Identification and sign saturation}

\bigskip{}

For bounded and convex $\Zcal$ with non-empty interior:

{\centering

\begin{tabular}{c|cc}
\multicolumn{3}{c}{}\tabularnewline
 & Sign saturation holds & Sign saturation does not hold\tabularnewline
\hline 
 &  & \tabularnewline
$\dot{G}(\cdot)$ non-periodic & \begin{tabular}{@{}c@{}}ID holds at every point\\ (Theorem \ref{thm: mixed reg})\end{tabular} & \begin{tabular}{@{}c@{}}ID fails at every point\\ (Theorem \ref{thm: neccessity part 1})\end{tabular}\tabularnewline
 &  & \tabularnewline
 &  & \tabularnewline
\begin{tabular}{@{}c@{}}$\dot{G}(\cdot)$ periodic \\ but non-constant\end{tabular} & \begin{tabular}{@{}c@{}}ID holds at every point\\ (Theorem \ref{thm: mixed reg})\end{tabular} & \begin{tabular}{@{}c@{}}ID holds in an open set \\(Theorem \ref{thm: ID period fun})\\ but fails in another open set \\(\cite{chamberlain2010binary}) \end{tabular}\tabularnewline
 &  & \tabularnewline
 &  & \tabularnewline
\begin{tabular}{@{}c@{}}$\dot{G}(\cdot)$ constant \\(logistic $F$)\end{tabular} & \begin{tabular}{@{}c@{}}ID holds at every point\\ (\cite{chamberlain1980analysis}\\ \cite{DaveziesID2021})\end{tabular} & \begin{tabular}{@{}c@{}}ID holds at every point\\ (\cite{chamberlain1980analysis}\\ \cite{DaveziesID2021})\end{tabular}\tabularnewline
 &  & \tabularnewline
\hline 
\multicolumn{3}{c}{}\tabularnewline
\end{tabular}

}
\end{table}

Exploiting the periodicity of $\dot{G}(\cdot)$ can be seen as a generalization
of the conditional maximum likelihood estimator for logistic distributions.
To see this, we notice that 
\[
\frac{P(Y_{1}=1,Y_{0}=0\mid X,\alpha)}{P(Y_{1}=0,Y_{1}=1\mid X,\alpha)}=\exp\left[G(W'\beta+\alpha)-G(\alpha)\right].
\]
Under logistic errors, $\dot{G}(\cdot)$ is a constant and thus $G(W'\beta+\alpha)-G(\alpha)$
does not depend on $\alpha$ for any $W$. If $\dot{G}(\cdot)$ is
a periodic function with the smallest positive period $\eta>0$, then
$G(W'\beta+\alpha)-G(\alpha)$ also does not depend on $\alpha$ whenever
$W'\beta/\eta$ is an integer. Since $W=(Z',1)'$, we only need to
have an interior point $z$ such that $(z',1)'\beta/\eta$ is an integer. 
\begin{rem}
As we have seen, the sign saturation condition guarantees the identification
of $\beta$ up to scale. Does it also guarantee the parametric rate
for estimation? The answer is no because Theorem 2 of \citet{chamberlain2010binary}
shows that the information bound is always zero outside the logistic
case, regardless of the identification status. This highlights a key
difference between identification and the parametric rate in estimation.
The question of identification depends on the sign saturation condition,
whereas the root-$n$ estimability depends on the semiparametric efficiency
calculation or the functional projection argument developed by \citet{bonhomme2012functional}.
\begin{comment}
The above analysis shows that distributions in the special class of
periodic $\dot{G}(\cdot)$ can have sufficient identification power
without the sign saturation condition. One natural question is whether
this special class of periodic $\dot{G}(\cdot)$ also has special
properties for estimation. One thing we can say is that the parametric
rate is still not possible outside the logistic case. This is because
Theorem 2 of \citet{chamberlain2010binary} still applies. Of course,
periodicity of $\dot{G}(\cdot)$ might provide other benefits in terms
of estimation even though the root-$n$ rate is impossible. We will
leave the exploration of this direction to future research. 
\end{comment}
\end{rem}

\section{Implications on marginal effects}

The identification of $\beta$ is directly related to the identification
of marginal effects in panel models with binary responses. The bounds
on the marginal effects often assume point identification of $\beta$,
see e.g., \citet*{DaveziesID2021}, \citet{liu2105.12891} and Theorem
6 of \citet*{chernozhukov2013average}. 

Here, we explore an important link between $\beta$ and the treatment
effects through the sign of components of $\beta$. Let us consider
the setting of Section \ref{subsec: necessary cond}: (1) the errors
$u_{t}$'s are i.i.d across $t$ with distribution $F(\cdot)$ and
are independent of $(X,\alpha)$ and (2) $X_{2,t}$ is binary with
$X_{2,t}=\oneb\{t=1\}$. We can interpret $X_{2,t}$ as a binary treatment;
in period $t=0$, no one is treated and in period $t=1$, everyone
is treated. For $\beta=(\beta_{1}',\beta_{2})'$, the sign of $\beta_{2}$
corresponds to the sign of common measures of treatment effects. For
example, the average partial effect effect at $X_{1,t}=z$ is 
\[
\Delta_{APE}(z)=P(z'\beta_{1}+\beta_{2}+\alpha\geq u_{t})-P(z'\beta_{1}+\alpha\geq u_{t}),
\]
see e.g., \citet{chernozhukov2013average}.

Although the magnitude of average partial effect is often not point
identified (see e.g., \citet{honore2006bounds} and \citet{chernozhukov2013average}),
there is hope that the sign of the effect is point identified, which
means that the identified set contains only positive numbers, only
negative numbers or only zero. Under Assumption \ref{assu: stationary error},
the support of $u_{t}$ is $\RR$, which means that ${\rm sgn}(\Delta_{APE}(z))={\rm sgn}(\beta_{2})$
for any $z$. Hence, identifying ${\rm sgn}(\Delta_{APE}(z))$ is
equivalent to identifying ${\rm sgn}(\beta_{2})$. In the literature,
there are also terms such as average treatment effects or ceteris
paribus effects, e.g., \citet{hoderlein2012nonparametric} and \citet{chernozhukov2015nonparametric}.
Under the assumptions of a linear index, the sign of these other measures
of treatment effects is also the sign of $\beta_{2}$. 
\begin{rem}
There are alternative ways of defining treatment effects but due to
the single-index nature of the model, the sign of other notions of
treatment effects is still related to ${\rm sgn}(\beta_{2})$. For
example, consider the individual effect $\oneb\{X_{1,t}\beta_{1}+\beta_{2}+\alpha\geq u_{t}\}-\oneb\{X_{1,t}\beta_{1}+\alpha\geq u_{t}\}$.
Clearly, this effect is negative with zero probability if and only
if $\beta_{2}>0$. 
\end{rem}
We now show that the sign saturation condition is necessary for this
purpose. To formally state the result, we rephrase Definition \ref{def: chamberlain def}. 
\begin{defn}
We say that $\beta$ and $b$ are observationally equivalent if for
any $x$ in the support, there exist $\pi_{x},\tilde{\pi}_{x}\in\Pi$
depending on $x$ such that $\int L(x;\beta,\alpha)d\pi_{x}(\alpha)=\int L(x;b,\alpha)d\tilde{\pi}_{x}(\alpha)$,
where $L$ is defined in (\ref{eq: L fun}).
\end{defn}
Recall that $\Bcal_{+}=\{\beta\in\RR^{K}:\ w'\beta>0\ \forall w\in\Wcal\}$
is the set of parameter values that do not satisfy the sign saturation
condition ($E(Y_{1}-Y_{0}\mid X)$ always positive). Consider two
subsets $\Bcal_{+,0}=\{\beta=(\beta_{1}',\beta_{2})'\in\Bcal_{+}:\ \beta_{2}=0\}$
and $\Bcal_{+,+}=\{\beta=(\beta_{1}',\beta_{2})'\in\Bcal_{+}:\ \beta_{2}>0\}$,
which denote the set of parameter values that imply zero effects and
positive effects of $X_{2,t}$, respectively. The following result
states the identification failure of ${\rm sgn}(\beta_{2})$. 
\begin{thm}
\label{thm: neccessity sign}Suppose that $\Zcal$ is bounded. Suppose
that $F(\cdot)$ is continuously differentiable such that $\dot{G}(\cdot)$
is not a periodic function. Then every point in $\Bcal_{+,0}$ is
observationally equivalent to a point in $\Bcal_{+,+}$. 
\end{thm}
By Theorem \ref{thm: mixed reg}, the sign saturation condition guarantees
point identification of $\beta$ up to scale and thus point identification
of ${\rm sgn}(\beta_{2})$. By Theorem \ref{thm: neccessity sign},
when the sign saturation condition fails, we cannot distinguish between
$\Delta_{APE}(z)=0$ and $\Delta_{APE}(z)>0$ unless $\dot{G}(\cdot)$
is periodic. Therefore, unless $\dot{G}(\cdot)$ is periodic, the
sign saturation condition is sufficient and necessary to guarantee
point identification of the sign of the marginal effects of $X_{2,t}$.
{} 

\section{\label{sec: test}Assessing the sign saturation condition}

We have seen that the sign saturation condition is sufficient and
necessary for the identification. The conditional mean function $\phi(X)=E(Y_{1}-Y_{0}\mid X)$
is identified. Therefore, ideally the sign saturation condition can
be checked in data. Here, we provide a simple check that does not
involve nonparametric estimation of $\phi$.

By Lemma \ref{lem: manski lem 1}, $P(\phi(X)\leq0)=1$ if and only
if $P(W'\beta\leq0)=1$. Hence, we define $\rho(q)=E\oneb\{W'q\geq0\}(Y_{1}-Y_{0})$
for $q\in\RR^{K}$. Here, we notice that we can replace $\RR^{K}$
with $[-1,1]^{K}$ or any set that contains an open neighborhood of
zero. It turns out that the sign saturation condition can be written
in terms of $\rho(\cdot)$ once we observe 
\[
\rho(\beta)=E\left(\oneb\{\phi(X)\geq0\}\cdot\phi(X)\right)\quad{\rm and}\quad\rho(-\beta)=E\left(\oneb\{\phi(X)\leq0\}\cdot\phi(X)\right).
\]

We give the formal statement below. 
\begin{lem}
\label{lem: test}Let Assumption \ref{assu: stationary error} hold.
Then
\[
\sup_{q\in\RR^{K}}\rho(q)=E\max\left\{ \phi(X),0\right\} \ {\rm and}\ \inf_{q\in\RR^{K}}\rho(q)=E\min\left\{ \phi(X),0\right\} .
\]
 
\end{lem}
Define $\tau_{*}=\min\{\tau_{1},-\tau_{2}\}$ with $\tau_{1}=\sup_{q\in\RR^{K}}\rho(q)$
and $\tau_{2}=\inf_{q\in\RR^{K}}\rho(q)$. By Lemma \ref{lem: test},
measuring $\tau_{*}$ can give us some indication of whether (or how
``well'') the sign saturation condition holds. In particular, the
sign saturation fails if and only if $\tau_{*}=0$. From the data,
we can construct a one-sided confidence interval for $\tau_{*}$.
The natural estimate for $\tau_{*}$ is $\hat{\tau}_{*}=\min\{\hat{\tau}_{1},-\hat{\tau}_{2}\}$,
where $\hat{\tau}_{1}=\sup_{q\in\RR^{K}}\hat{\rho}_{n}(q)$, $\hat{\tau}_{2}=\inf_{q\in\RR^{K}}\hat{\rho}_{n}(q)$
and
\[
\hrho_{n}(q)=n^{-1}\sum_{i=1}^{n}\oneb\{W_{i}'q\geq0\}(Y_{i,1}-Y_{i,0}).
\]

In the proof of Theorem \ref{thm: test} below, we show that
\[
\sqrt{n}\left(\hat{\tau}_{*}-\tau_{*}\right)\leq\left(\sup_{q\in\RR^{K}}S_{n}(q)\right)\cdot\oneb\{\tau_{1}\leq-\tau_{2}\}+\left(\sup_{q\in\RR^{K}}(-S_{n}(q))\right)\cdot\oneb\{\tau_{1}>-\tau_{2}\},
\]
where $S_{n}(q)=\sqrt{n}(\hat{\rho}_{n}(q)-\rho(q))$. By the standard
arguments of empirical processes, $S_{n}$ converges to a mean-zero
Gaussian process, which means that $\sup_{q\in\RR^{K}}S_{n}(q)$ and
$\sup_{q\in\RR^{K}}(-S_{n}(q))$ have the same distribution. Thus,
a simple $(1-\alpha)$-confidence interval for $\tau_{*}$ can be
obtained once we approximate the $(1-\alpha)$ quantile of $\sup_{q\in\RR^{K}}S_{n}(q)$.
This motivates the following bootstrapping algorithm. 
\begin{lyxalgorithm}
\label{alg: test }For a $(1-\alpha)$-confidence interval for $\tau_{*}$:
\begin{enumerate}
\item Collect data $\{(W_{i},Y_{i,1},Y_{i,0})\}_{i=1}^{n}$.
\item Compute the estimate $\hat{\tau}_{*}=\min\{\hat{\tau}_{1},-\hat{\tau}_{2}\}$. 
\item Draw a random sample $\{(\tilde{W}_{i},\tilde{Y}_{i,1},\tilde{Y}_{i,0})\}_{i=1}^{n}$
with replacement from the data and compute $\sup_{q\in\RR^{K}}\tilde{S}_{n}(q)$,
where $\tilde{S}_{n}(q)=\sqrt{n}\left(n^{-1}\sum_{i=1}^{n}\oneb\{\tilde{W}_{i}'q>0\}(\tilde{Y}_{i,1}-\tilde{Y}_{i,0})-\hrho_{n}(q)\right)$. 
\item Repeat the previous step many times and compute $c_{1-\alpha}$, the
$(1-\alpha)$ quantile of $\sup_{q\in\RR^{K}}\tilde{S}_{n}(q)$.\footnote{In practice, optimization over $q\in\RR^{K}$ can be done by a grid
search over $N$ points on the unit sphere $\{v\in\RR^{K}:\ \|v\|_{2}=1\}$.
The computation seems reasonably fast for a normal sample size. For
example, in the setting of Section \ref{sec: monte carlo} with $\beta_{2}=1$,
$n=5000$ and grid size $N=2000$, computing $\sup_{q}\tilde{S}_{n}(q)$
with 500 bootstrap samples take less than 17 seconds in Matlab on
a 2022 MacBook Pro laptop. Further speedups are possible depending
on $n$, $N$ and available memory.} 
\item The $(1-\alpha)$-confidence interval for $\tau_{*}$ is $\left[\max(\hat{\tau}_{*}-c_{1-\alpha}n^{-1/2},0),\ 1\right]$. 
\end{enumerate}
\end{lyxalgorithm}
Notice that computing $\hat{\tau}_{1}$ is equivalent to computing
the maximum score estimator\footnote{Notice that maximizing $\hat{\rho}_{n}(q)$ is equivalent to maximizing
$n^{-1}\sum_{i=1}^{n}{\rm sgn}(W_{i}'q)(Y_{i,1}-Y_{i,0})$ because
${\rm sgn}(W_{i}'q)=-1+2\cdot\oneb\{W_{i}'q\geq0\}$.} and all the existing computational algorithms and software for the
maximum score estimator can be used. Finding $\hat{\tau}_{2}$ also
reduces to computing the maximum score estimator once we swap $Y_{1}$
and $Y_{0}$. The above algorithm can be justified by the following
result.
\begin{thm}
\label{thm: test}Let Assumption \ref{lem: manski lem 1} hold. Assume
that $P(Y_{1}=Y_{0})<1$. Then
\[
\limsup_{n\rightarrow\infty}P\left(\sqrt{n}(\hat{\tau}_{*}-\tau_{*})>c_{1-\alpha}\right)\leq\alpha.
\]
\end{thm}
Note that the cube-root asymptotics typically associated with the
maximum score estimator (see e.g., \citet{kim1990cube} and \citet{seo2018local})
does not arise here. The reason is that $\hat{\tau}_{1}$ is the maximum
of $\hat{\rho}_{n}(\cdot)$, rather than the argmax. The cube-root
asymptotics of the maximum score estimator is due to the non-standard
rate of some terms in the expansion for analyzing the argmax. Fortunately,
we do not have to deal with such terms for our purpose. 

It is worth pointing out that Algorithm \ref{alg: test } is asymptotically
exact in the sense that there exists a data-generating process under
which the asymptotic coverage probability is exactly $\alpha$. For
example, if $\rho(q)=0$ for any $q$, then $\sqrt{n}(\hat{\tau}_{*}-\tau_{*})=\sup_{q\in\RR^{K}}S_{n}(q)$
and $P\left(\sqrt{n}(\hat{\tau}_{*}-\tau_{*})>c_{1-\alpha}\right)\rightarrow\alpha$.

Our results characterize what drives the identification and introduce
$\tau_{*}$ as a measure for the identification strength. This theoretical
insight motivates a preliminary diagnostic step in applied work. Reporting
the confidence interval of $\tau_{*}$ in Algorithm \ref{alg: test }
can offer important insight on the identification strength. For example,
if this confidence interval contains zero, then one should be concerned
about identification failure. 

However, testing $\tau_{*}=0$ might not detect all the problematic
cases. In practice, the identification might not be clear-cut and
additional caution is warranted if we worry about the ``weak'' identification
scenario, which can complicate econometric analysis due to the generic
issue of post-selection inference, see e.g., \citet{Leeb2005}. 

Fortunately, since $\tau_{*}$ can be learned from the data, one option
is to proceed with the identification-based inference (i.e., cube-root
asymptotics) only when the identification is strong enough. For instance,
one might adopt a robust default inference method and switch to the
cube-root approach only when a test fails to reject $H_{0}:\ \tau_{*}\geq\tau_{0}$,
where $\tau_{0}>0$ is a pre-specified fixed threshold. This is uniformly
valid asymptotically as it rules out $\tau_{*}$ being ``local''
to zero. 

This strategy is conceptually analogous to a common treatment of weak
instruments in the linear instrumental variable (IV) models: low correlation
between the instrument and the exogenous variable suggests potential
identification failure and one solution is to proceed with the classical
asymptotics only when this correlation is large enough, e.g., measured
by the first-stage $F$-statistic \citep{stock2002testing}.

This raises a natural question: what should this default identification-robust
method be? In the IV literature, approaches such as the Anderson-Rubin
test have been developed. In our model, one could consider inverting
the maximum score criterion function, in a manner similar to the Anderson--Rubin
approach. While more refined methods may be possible, they would require
substantial additional analysis on asymptotic theory, which lies beyond
the current focus on characterizing identification and are left for
future research.%
{} %
\begin{comment}
For the weak IV problem, one solution is to check whether the IV is
strong enough and proceed accordingly, such as Stock and Yogo (2005)??.
In both problems, the identification strength can be assessed from
the data. In our model, we can learn $\tau_{*}$; in the IV model,
we can learn the correlation between the IV and the exogenous variable.
We recommend the following approach for inference: start with an identification-robust
method for inference, check the identification strength and only switch
to the cubic-root asymptotics when we are satisfied with strong identification. 
\end{comment}
\begin{comment}
\textbf{ADD weak identification discussion.} We should set the threshold
to be high enough. 
\end{comment}
\begin{comment}
this approach does not provide a complete answer. What should we do
when the identification is not strong enough?
\end{comment}
\begin{comment}
For example, suppose that we have an identification-robust inference
method as a default method. We can set a threshold $\tau_{0}$ and
only switch from the default method to a method based on cube-root
asymptotics when we fail to reject
\end{comment}

\section{\label{sec: monte carlo}Numerical illustration}

We have seen that $\tau_{*}$ is a measure of the sign saturation
condition and thus serves a gauge for identification. Using Monte
Carlo simulations, we now illustrate how $\tau_{*}$ relates to the
identification strength and the accuracy of the cube-root asymptotics.

Consider the model in (\ref{eq: FE model}) with $X_{t}=(\oneb\{t=0\},X_{2,t})'$
for $t\in\{0,1\}$, where $X_{2,t}$ is from the uniform distribution
on $[-1,1]$, $u_{t}$ is from the standard normal distribution and
$\alpha_{t}=(X_{2,0}+X_{2,1})/2$. Here, $X_{2,0}$, $X_{2,1}$, $u_{0}$
and $u_{1}$ are mutually independent. We set $\beta=(1,\beta_{2})'$.
Therefore, $W'\beta=1+\beta_{2}(X_{2,1}-X_{2,0})$. Clearly, the support
of $W'\beta$ is $[-2|\beta_{2}|+1,2|\beta_{2}|+1]$, which means
that the sign saturation holds if and only if $|\beta_{2}|>0.5$.
In Figure \ref{fig: ID strength}, we plot $\tau_{*}$ (computed using
$\hat{\tau}_{*}$ with $n=10^{8}$) as a function of $\beta_{2}$.
It confirms the intuition that $\beta_{2}=0.5$ corresponds to $\tau_{*}=0$,
which means that the sign saturation condition fails. Higher value
of $\beta_{2}$ corresponds to a higher degree to which the sign saturation
condition holds. 

\begin{figure}
\caption{\label{fig: ID strength}Coefficient and identification}

\centering{}\includegraphics[scale=0.6]{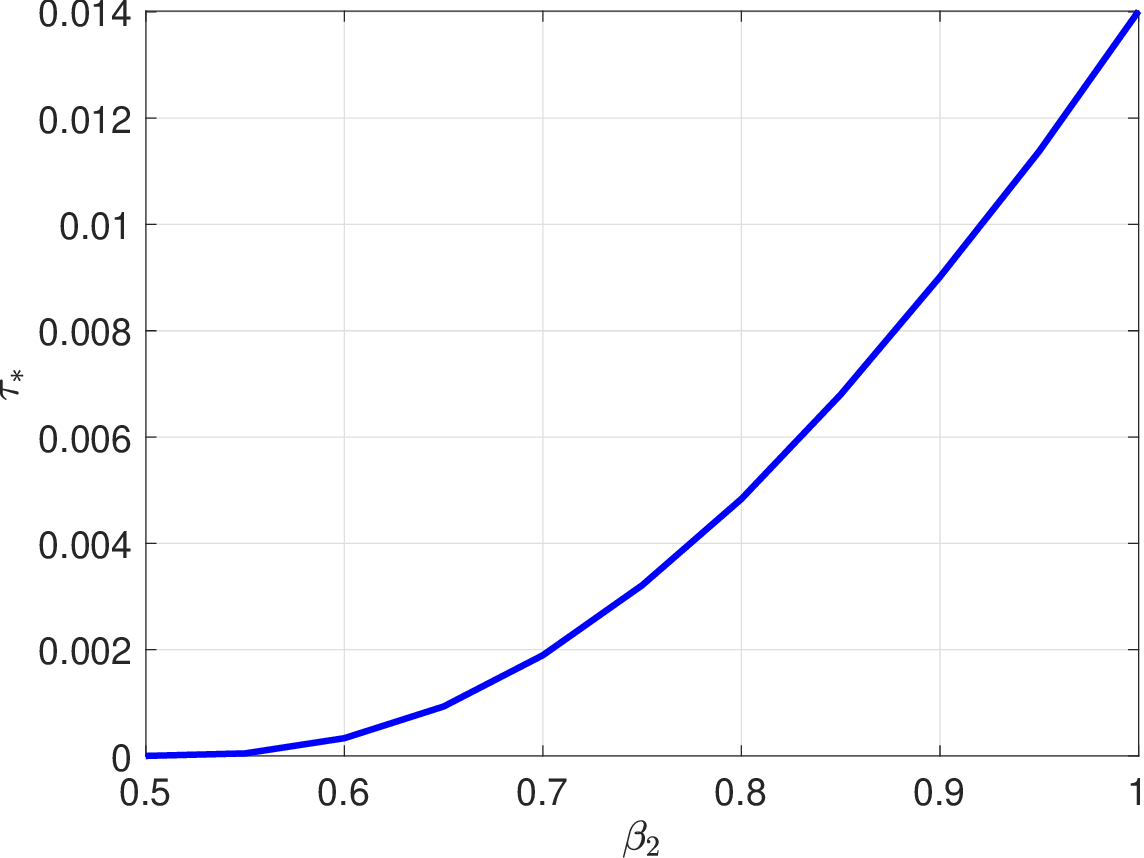}
\end{figure}

When the identification holds (together with other regularity conditions),
one can expect the cube-root asymptotics. Given a particular value
of $\tau_{*}$, we compare the finite-sample distribution of the maximum
score estimator with its asymptotic distribution. Explicitly deriving
the cube-root asymptotic distribution requires quantities that are
difficult to compute. To circumvent this problem, we treat the distribution
of the estimator with a sample size of $n=1,000,000$ as the asymptotic
distribution.\footnote{Here, the goal is to assess the importance of identification by examining
how well the asymptotic theory applies. In practice, there is another
issue of approximating the asymptotic distribution either with sub-sampling
or with a modified bootstrap.} The finite-sample distribution is computed using $n=1,000$. To make
these two distributions comparable, we compare the maximum score estimator
$\hat{\beta}_{2}$ and consider the distribution of $n^{1/3}(\hat{\beta}_{2}-\beta_{2})$,
which should converge to a limiting distribution under the cube-root
asymptotics. We present the comparison in Figure \ref{fig: dist}
for $\beta_{2}\in\{0.6,\ 1\}$ based on 5,000 repetitions. We see
that when the identification is stronger ($\beta_{2}=1$) and the
finite-sample distribution of the maximum score estimator is better
approximated by the asymptotic distribution. %

\begin{figure}
\caption{\label{fig: dist}Distribution of $n^{1/3}(\hat{\beta}_{2}-\beta_{2})$
for $\beta_{2}=0.6$ (left) and $\beta_{2}=1$ (right)}

\centering{}\includegraphics[scale=0.4]{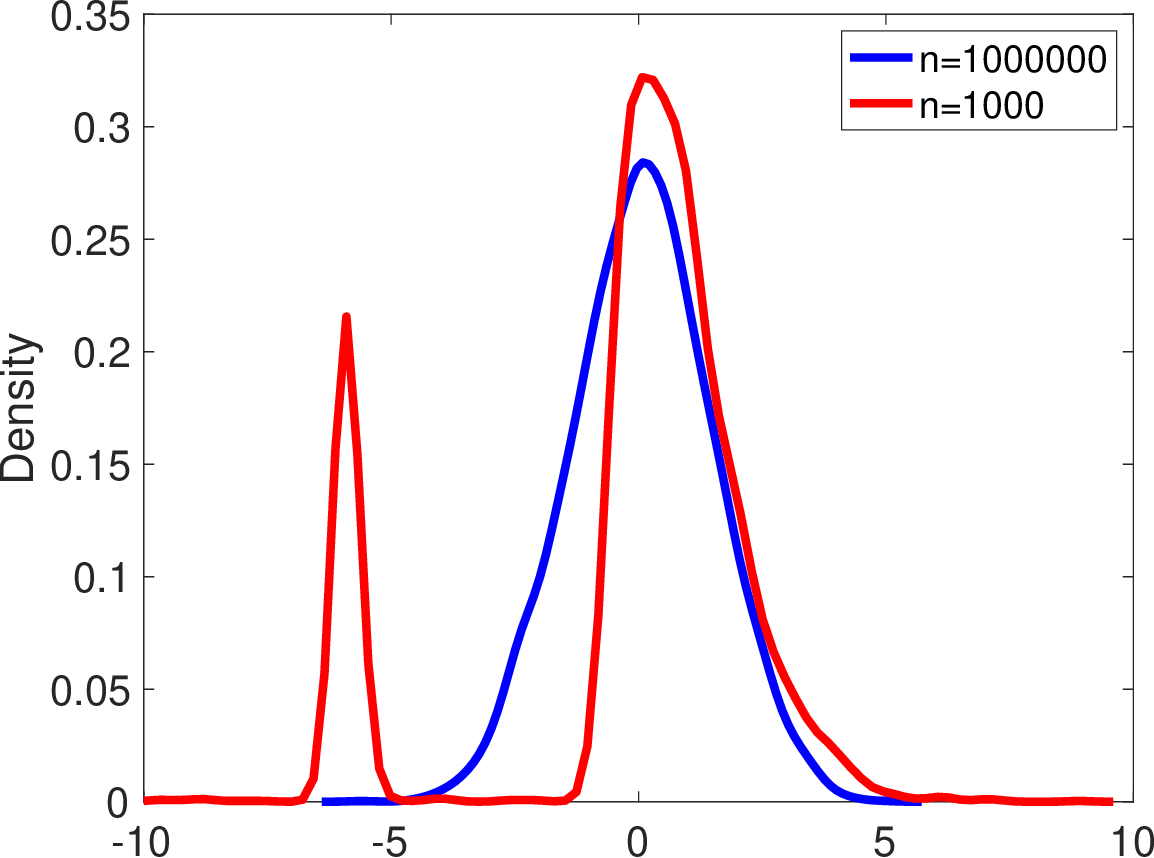}\includegraphics[scale=0.4]{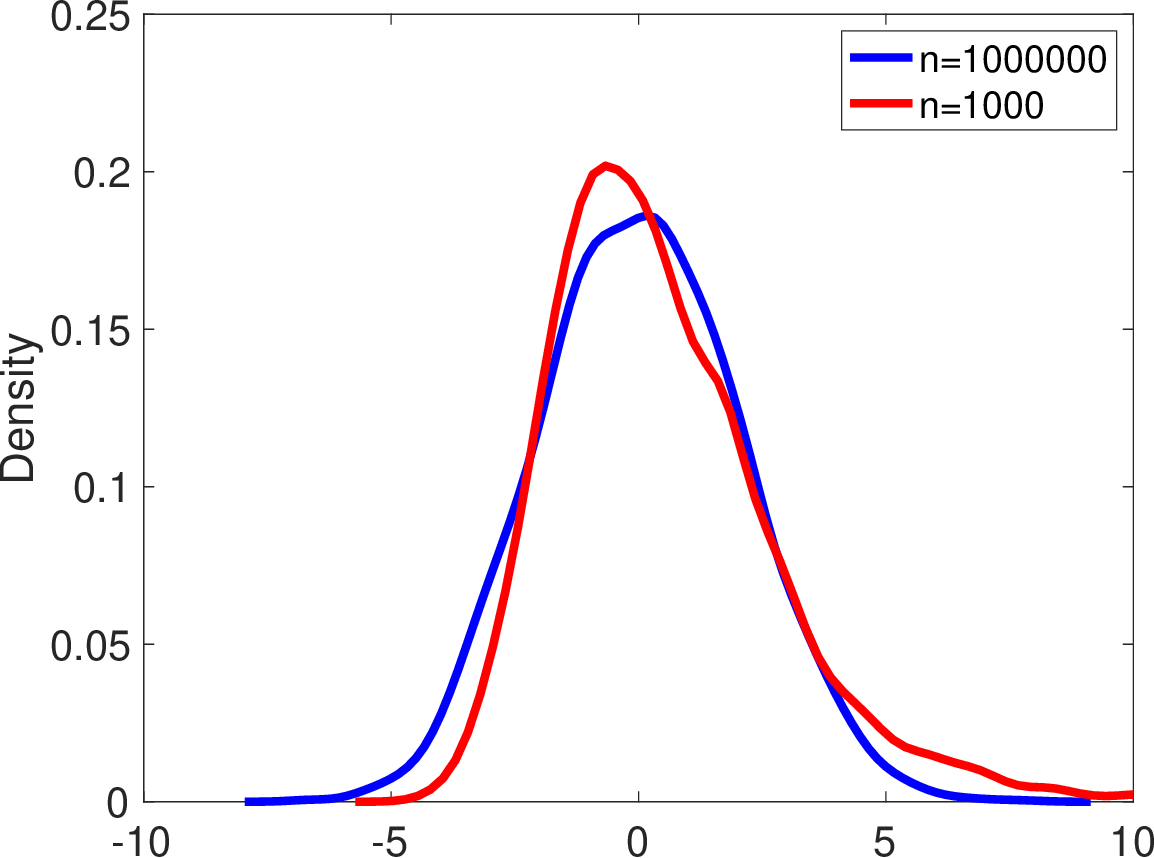}
\end{figure}

\appendix

\section{Appendix: proofs of theoretical results}

\subsection{Proof of Theorem \ref{thm: mixed reg}}

We define the following notations that will be used throughout this
subsection (proof of Theorem \ref{thm: mixed reg}). For any set $A$,
$\overline{A}$ denotes the closure of $A$ and $A^{\circ}$ denotes
the interior of $A$. Let $\Zcal$ be the support of $Z$. We partition
$\beta=(\beta_{1}',\beta_{2})'$ and $b=(b_{1}',b_{2})'$ with $\beta_{1},b_{1}\in\RR^{K-1}$
and $\beta_{2},b_{2}\in\RR$. Then $W'\beta=Z'\beta_{1}+\beta_{2}$.
Define the following sets: $A_{1}=\{z\in\RR^{K-1}:\ z'\beta_{1}+\beta_{2}>0\}$,
$A_{2}=\{z\in\RR^{K-1}:\ z'b_{1}+b_{2}<0\}$, $B_{1}=\{z\in\RR^{K-1}:\ z'\beta_{1}+\beta_{2}<0\}$
and $B_{2}=\{z\in\RR^{K-1}:\ z'b_{1}+b_{2}>0\}$. Before we prove
Theorem \ref{thm: mixed reg}, we establish some auxiliary results.
\begin{lem}
\label{lem: topology fact 1} Suppose that $A$ is an open set (in
$\RR^{K-1}$) such that $A\bigcap\Zcal\neq\emptyset$. Then $P(Z\in A)>0$.
\end{lem}
\begin{proof}
Recall that $\Zcal$ is the smallest closed set such that $P(Z\in\Zcal)=1$,
see the definition of the support of a measure on page 227 of \citet{dudley_2002}. 

Suppose that $P(Z\in A)=0$. Then $P(Z\in\Zcal\backslash A)=1$. Since
$A$ is open, $\Zcal\backslash A$ is a closed set. Notice that $\Zcal\backslash A$
is a proper subset of $\Zcal$ because $A\bigcap\Zcal$ is non-empty.
Then $\Zcal\backslash A$ is a closed set with full probability mass
and is a smaller set than $\Zcal$, contradicting the assumption that
$\Zcal$ is the support of $Z$.
\end{proof}
\begin{lem}
\label{lem: topology fact 2} Suppose that $B$ is an open set (in
$\RR^{K-1}$) such that $\overline{B}\bigcap\Zcal^{\circ}\neq\emptyset$.
Then $P(Z\in B)>0$.
\end{lem}
\begin{proof}
Let $A=B\bigcap\Zcal^{\circ}$. By the openness of $B$ and $\Zcal^{\circ}$,
$A$ is open. The desired result follows by Lemma \ref{lem: topology fact 1}
once we show that $A$ is non-empty. 

We do so by contradiction. Suppose $B\bigcap\Zcal^{\circ}=\emptyset$.
Then $B\subseteq(\Zcal^{\circ})^{c}$. Since $\Zcal^{\circ}$ is open,
$(\Zcal^{\circ})^{c}$ is closed. Thus, $(\Zcal^{\circ})^{c}$ is
a closed set containing $B$. Recall that $\overline{B}$ is the smallest
closed set containing $B$. It follows that $\overline{B}\subseteq(\Zcal^{\circ})^{c}$,
which means that $\overline{B}\bigcap\Zcal^{\circ}=\emptyset$. This
contradicts the assumption of $\overline{B}\bigcap\Zcal^{\circ}\neq\emptyset$.
Hence, $A=B\bigcap\Zcal^{\circ}$ is non-empty. 
\end{proof}
\begin{lem}
\label{lem: topology fact 3}Suppose that $b\neq\mu\beta$ for any
$\mu>0$. If $P(Z\in A_{1})$, $P(Z\in A_{2})$, $P(Z\in B_{1})$
and $P(Z\in B_{2})$ are all strictly positive, then $\overline{A_{1}\bigcap A_{2}}=\overline{A}_{1}\bigcap\overline{A}_{2}$
and $\overline{B_{1}\bigcap B_{2}}=\overline{B}_{1}\bigcap\overline{B}_{2}$. 
\end{lem}
\begin{proof}
We prove the first claim; the second claims follows by an analogous
argument. 

Since $P(Z\in A_{1})$ and $P(Z\in B_{1})$ are both strictly positive,
we have that $\beta_{1}\neq0$. To see this, observe that if $\beta_{1}=0$,
then the two events $\{Z\in A_{1}\}=\{\beta_{2}>0\}$ and $\{Z\in B_{1}\}=\{\beta_{2}<0\}$
cannot both have strictly positive probability. Similarly, $P(Z\in A_{2})>0$
and $P(Z\in B_{2})$ yield yield $b_{1}\neq0$.

First observe that $\overline{A_{1}\bigcap A_{2}}\subseteq\overline{A}_{1}\bigcap\overline{A}_{2}$
(because $A_{1}\bigcap A_{2}\subseteq\overline{A}_{1}\bigcap\overline{A}_{2}$
and $\overline{A}_{1}\bigcap\overline{A}_{2}$ is closed). 

To show the other direction, we fix an arbitrary $z_{*}\in\overline{A}_{1}\bigcap\overline{A}_{2}$.
We need to show that $z_{*}\in\overline{A_{1}\bigcap A_{2}}$. 

Notice that $z_{*}'\beta_{1}+\beta_{2}\geq0$ (due to $z_{*}\in\overline{A}_{1})$
and $z_{*}'b_{1}+b_{2}\leq0$ (due to $z_{*}\in\overline{A}_{2})$.
If $z_{*}'\beta_{1}+\beta_{2}>0$ and $z_{*}'b_{1}+b_{2}<0$, then
$z_{*}\in A_{1}\bigcap A_{2}$, which means that $z_{*}\in\overline{A_{1}\bigcap A_{2}}$.
Therefore, we only need to show that $z_{*}\in\overline{A_{1}\bigcap A_{2}}$
in the following three cases:
\begin{itemize}
\item[(i)] $z_{*}'\beta_{1}+\beta_{2}=0$ and $z_{*}'b_{1}+b_{2}<0$
\item[(ii)] $z_{*}'\beta_{1}+\beta_{2}>0$ and $z_{*}'b_{1}+b_{2}=0$
\item[(iii)] $z_{*}'\beta_{1}+\beta_{2}=0$ and $z_{*}'b_{1}+b_{2}=0$.
\end{itemize}
Consider Case (i). Let $z_{j}=z_{*}+cj^{-1}\beta_{1}$ for $j\in\{1,2,...\}$.
Since $z_{*}'b_{1}+b_{2}<0$, we can choose a constant $c>0$ to be
small enough such that $cb_{1}'\beta_{1}<|z_{*}'b_{1}+b_{2}|/2$,
which would imply that $z_{j}'b_{1}+b_{2}<(z_{*}'b_{1}+b_{2})/2<0$.
On the other hand, notice that $z_{j}'\beta_{1}+\beta_{2}=(z_{*}'\beta_{1}+\beta_{2})+cj^{-1}\|\beta_{1}\|_{2}^{2}=cj^{-1}\|\beta_{1}\|_{2}^{2}>0$
due to $\beta_{1}\neq0$. Thus, we have constructed a sequence $\{z_{j}\}_{j=1}^{\infty}$
in $A_{1}\bigcap A_{2}$ such that $\lim_{j\rightarrow\infty}z_{j}=z_{*}$.
Hence, $z_{*}$ is a limit point of $A_{1}\bigcap A_{2}$, which means
that $z_{*}\in\overline{A_{1}\bigcap A_{2}}$. 

Similarly, we can show that $z_{*}\in\overline{A_{1}\bigcap A_{2}}$
in Case (ii). 

We now consider Case (iii). We first show that there exists $\Delta\in\RR^{K-1}$
such that $\Delta'\beta_{1}>0$ and $\Delta'b_{1}<0$. We proceed
by contradiction. Suppose that there does not exist $\Delta\in\RR^{K-1}$
with $\Delta'\beta_{1}>0$ and $\Delta'b_{1}<0$. Then for any $\Delta\in\RR^{K-1}$,
$\Delta'\beta_{1}>0$ implies $\Delta'b_{1}\geq0$. By the continuity
of $\Delta\mapsto\Delta'\beta_{1}$, $\Delta'\beta_{1}\geq0$ implies
$\Delta'b_{1}\geq0$. By Farkas's lemma (e.g., Corollary 22.3.1 of
\citet{rockafellar1970convex}), there exists $\xi\geq0$ such that
$b_{1}=\xi\beta_{1}$. Since we have proved $b_{1}\neq0$, it follows
that $\xi\neq0$, which means that $\xi>0$. Now we have $z_{*}'b_{1}+b_{2}=\xi z_{*}'\beta_{1}+b_{2}=\xi(z_{*}'\beta_{1}+b_{2}\xi^{-1})$.
Since $z_{*}'\beta_{1}+\beta_{2}=0$ and $z_{*}'b_{1}+b_{2}=0$, it
follows that $b_{2}\xi^{-1}=\beta_{2}$. Therefore, we have $b=(b_{1}',b_{2})'=\xi(\beta_{1}',\beta_{2})'$,
contradicting the assumption $b\neq\mu\beta$ for any $\mu>0$. Hence,
there exists $\Delta\in\RR^{K-1}$ such that $\Delta'\beta_{1}>0$
and $\Delta'b_{1}<0$. 

Now we choose $\Delta\in\RR^{K-1}$ such that $\Delta'\beta_{1}>0$
and $\Delta'b_{1}<0$. Define $\tilde{z}_{j}=z_{*}+j^{-1}\Delta$
for $j\in\{1,2,...\}$. Then $\tilde{z}_{j}'\beta_{1}+\beta_{2}=(z_{*}'\beta_{1}+\beta_{2})+j^{-1}\Delta'\beta_{1}=j^{-1}\Delta'\beta_{1}>0$
and $\tilde{z}_{j}'b_{1}+b_{2}=(z_{*}'b_{1}+b_{2})+j^{-1}\Delta'b_{1}=j^{-1}\Delta'b_{1}<0$.
Again, we have constructed a sequence $\{\tilde{z}_{j}\}_{j=1}^{\infty}$
in $A_{1}\bigcap A_{2}$ such that $\lim_{j\rightarrow\infty}\tilde{z}_{j}=z_{*}$.
Hence, $z_{*}$ is a limit point of $A_{1}\bigcap A_{2}$, which means
that $z_{*}\in\overline{A_{1}\bigcap A_{2}}$. 
\end{proof}
\begin{lem}
\label{lem: topology fact 4}Suppose that $\Zcal$ is convex and $\Zcal^{\circ}$
is non-empty. Moreover, assume that $P(Z\in A_{1})$ and $P(Z\in A_{2})$
are both strictly positive. Then at least one of $\overline{A}_{1}\bigcap\overline{A}_{2}\bigcap\Zcal^{\circ}$
and $\overline{B}_{1}\bigcap\overline{B}_{2}\bigcap\Zcal^{\circ}$
is not empty. 
\end{lem}
\begin{proof}
We proceed by contradiction. Suppose that $\overline{A}_{1}\bigcap\overline{A}_{2}\bigcap\Zcal^{\circ}=\emptyset$
and $\overline{B}_{1}\bigcap\overline{B}_{2}\bigcap\Zcal^{\circ}=\emptyset$.
Clearly, $\overline{A}_{1}=\{z\in\RR^{K-1}:\ z'\beta_{1}+\beta_{2}\geq0\}$
and $\overline{A}_{2}=\{z\in\RR^{K-1}:\ z'b_{1}+b_{2}\leq0\}$ as
well as $\overline{B}_{1}=\{z\in\RR^{K-1}:\ z'\beta_{1}+\beta_{2}\leq0\}$
and $\overline{B}_{2}=\{z\in\RR^{K-1}:\ z'b_{1}+b_{2}\geq0\}$. 

Since $P(Z\in A_{1})$ and $P(Z\in A_{2})$ are strictly positive,
$A_{1}$ and $A_{2}$ are both non-empty. Then there exist $\delta_{1},\delta_{2}>0$
and $z_{1},z_{2}\in\Zcal$ such that $z_{1}'\beta_{1}+\beta_{2}=\delta_{1}$
and $z_{2}'b_{1}+b_{2}=-\delta_{2}$. Since $\Zcal$ is convex and
$\Zcal^{\circ}$ is non-empty, we have that the interior of $\Zcal$
is dense in $\overline{\Zcal}$, i.e., $\overline{\Zcal^{\circ}}=\overline{\Zcal}$;
see Lemma 5.28 of \citet{aliprantis2006infinite}. By definition,
$\Zcal$ is closed ($\overline{\Zcal}=\Zcal$), which means that $\overline{\Zcal^{\circ}}=\Zcal$.
Thus, we can approximate $z_{1}$ and $z_{2}$ by points in $\Zcal^{\circ}$.
Hence, we can change $\delta_{1},\delta_{2}>0$ slightly such that
$z_{1}'\beta_{1}+\beta_{2}=\delta_{1}$ and $z_{2}'b_{1}+b_{2}=-\delta_{2}$
with $z_{1},z_{2}\in\Zcal^{\circ}$.

By $\overline{A}_{1}\bigcap\overline{A}_{2}\bigcap\Zcal^{\circ}=\emptyset$,
we have $z_{1}'b_{1}+b_{2}>0$; since $z_{1}\in\overline{A}_{1}\bigcap\Zcal^{\circ}$,
$z_{1}'b_{1}+b_{2}\leq0$ would imply $z_{1}\in\overline{A}_{1}\bigcap\overline{A}_{2}\bigcap\Zcal^{\circ}$
and contradict $\overline{A}_{1}\bigcap\overline{A}_{2}\bigcap\Zcal^{\circ}=\emptyset$.
Similarly, $\overline{A}_{1}\bigcap\overline{A}_{2}\bigcap\Zcal^{\circ}=\emptyset$
implies $z_{2}'\beta_{1}+\beta_{2}<0$; since $z_{2}\in\overline{A}_{2}\bigcap\Zcal^{\circ}$,
$z_{2}'\beta_{1}+\beta_{2}\geq0$ would imply $z_{2}\in\overline{A}_{1}\bigcap\overline{A}_{2}\bigcap\Zcal^{\circ}$.

Now define $z_{*}(\lambda)=\lambda z_{1}+(1-\lambda)z_{2}$ for $\lambda\in(0,1)$.
Since $\Zcal$ is convex, we have that $\Zcal^{\circ}$ is also convex
by Lemma 5.27 of \citet{aliprantis2006infinite}. Thus, $z_{*}(\lambda)\in\Zcal^{\circ}$
for any $\lambda\in(0,1)$. By $\overline{A}_{1}\bigcap\overline{A}_{2}\bigcap\Zcal^{\circ}=\emptyset$
and $\overline{B}_{1}\bigcap\overline{B}_{2}\bigcap\Zcal^{\circ}=\emptyset$,
we have that $z_{*}(\lambda)\notin A_{1}\bigcap A_{2}$ and $z_{*}(\lambda)\notin B_{1}\bigcap B_{2}$
for any $\lambda\in(0,1)$. This means that $z_{*}(\lambda)'\beta_{1}+\beta_{2}$
and $z_{*}(\lambda)'b_{1}+b_{2}$ have the same sign. In other words,
we have that $f(\lambda)\geq0$ for any $\lambda\in(0,1)$, where
\[
f(\lambda)=\left[z_{*}(\lambda)'\beta_{1}+\beta_{2}\right]\cdot\left[z_{*}(\lambda)'b_{1}+b_{2}\right].
\]

By straight-forward algebra using $z_{1}'\beta_{1}+\beta_{2}=\delta_{1}$
and $z_{2}'b_{1}+b_{2}=-\delta_{2}$, we have 
\begin{equation}
z_{*}(\lambda)'\beta_{1}+\beta_{2}=\lambda\delta_{1}+(1-\lambda)\left(z_{2}'\beta_{1}+\beta_{2}\right)=\left[\delta_{1}-\left(z_{2}'\beta_{1}+\beta_{2}\right)\right]\lambda+\left(z_{2}'\beta_{1}+\beta_{2}\right)\label{eq: quad eq 1}
\end{equation}
and 
\begin{equation}
z_{*}(\lambda)'b_{1}+b_{2}=\lambda\left(z_{1}'b_{1}+b_{2}\right)-(1-\lambda)\delta_{2}=\left[\left(z_{1}'b_{1}+b_{2}\right)+\delta_{2}\right]\lambda-\delta_{2}.\label{eq: quad eq 2}
\end{equation}

Hence, $f(\lambda)$ is a quadratic function of $\lambda$. By $\delta_{1},\delta_{2}>0$,
$z_{1}'b_{1}+b_{2}>0$ and $z_{2}'\beta_{1}+\beta_{2}<0$, the coefficient
of the $\lambda^{2}$ term in $f(\lambda)$ satisfies
\[
\left[\delta_{1}-\left(z_{2}'\beta_{1}+\beta_{2}\right)\right]\cdot\left[\left(z_{1}'b_{1}+b_{2}\right)+\delta_{2}\right]>\delta_{1}\delta_{2}>0.
\]

Therefore, $f(\cdot)$ is strictly convex. By (\ref{eq: quad eq 1})
and (\ref{eq: quad eq 2}), we notice that $f(\lambda_{1})=f(\lambda_{2})=0$,
where 
\[
\lambda_{1}=-\frac{z_{2}'\beta_{1}+\beta_{2}}{\delta_{1}-\left(z_{2}'\beta_{1}+\beta_{2}\right)}\qquad{\rm and}\qquad\lambda_{2}=\frac{\delta_{2}}{\left(z_{1}'b_{1}+b_{2}\right)+\delta_{2}}.
\]

Since $\delta_{1},\delta_{2}>0$, $z_{1}'b_{1}+b_{2}>0$ and $z_{2}'\beta_{1}+\beta_{2}<0$,
we have $\lambda_{1}\in(0,1)$ and $\lambda_{2}\in(0,1)$. We notice
that we must have $\lambda_{1}=\lambda_{2}$. Otherwise, there would
exist $\lambda$ that is strictly between $\lambda_{1}$ and $\lambda_{2}$,
i.e., $\min\{\lambda_{1},\lambda_{2}\}<\lambda<\max\{\lambda_{1},\lambda_{2}\}$.
Since $f(\cdot)$ is strictly convex and $f(\lambda_{1})=f(\lambda_{2})=0$,
we have that $f(\lambda)<0$ for any $\lambda$ that is strictly between
$\lambda_{1}$ and $\lambda_{2}$, contradicting $f(\lambda)\geq0$
for all $\lambda\in(0,1)$. Hence, we must have that $\lambda_{1}=\lambda_{2}$. 

Let $\lambda_{0}$ denote this common value of $\lambda_{1}=\lambda_{2}$.
Then $z_{*}(\lambda_{0})'\beta_{1}+\beta_{2}=0$ and $z_{*}(\lambda_{0})'b_{1}+b_{2}=0$
from (\ref{eq: quad eq 1}) and (\ref{eq: quad eq 2}). Therefore,
we have found $z_{0}:=z_{*}(\lambda_{0})\in\Zcal^{\circ}$ such that
$z_{0}'\beta_{1}+\beta_{2}=0$ and $z_{0}'b_{1}+b_{2}=0$. In other
words, $z_{0}\in\overline{A}_{1}\bigcap\overline{A}_{2}\bigcap\Zcal^{\circ}$
and $z_{0}\in\overline{B}_{1}\bigcap\overline{B}_{2}\bigcap\Zcal^{\circ}$.
This contradicts the assumption that $\overline{A}_{1}\bigcap\overline{A}_{2}\bigcap\Zcal^{\circ}=\emptyset$
and $\overline{B}_{1}\bigcap\overline{B}_{2}\bigcap\Zcal^{\circ}=\emptyset$. 
\end{proof}
\begin{proof}[\textbf{Proof of Theorem \ref{thm: mixed reg}}]
By Lemma \ref{lem: manski lem 1}, Assumption \ref{assu: sign saturation}
(i.e., $P\left(E(Y_{1}-Y_{0}\mid X)>0\right)>0$ and $P\left(E(Y_{1}-Y_{0}\mid X)<0\right)>0$)
implies that $P(Z'\beta_{1}+\beta_{2}>0)>0$ and $P(Z'\beta_{1}+\beta_{2}<0)>0$. 

The rest of the proof proceeds in two steps, in which we prove the
``if'' and the ``only if'' parts, respectively.

\textbf{Step 1:} show that if $R(b)=0$, then $b=\mu\beta$ for some
$\mu>0$. 

We proceed by contradiction. Suppose that $R(b)=0$ and $b\neq\mu\beta$
for any $\mu>0$. Since $R(b)=0$, we have that ${\rm sgn}(E(Y_{1}-Y_{0}\mid X))={\rm sgn}(W'b)$
with probability one. By Lemma \ref{lem: manski lem 1}, ${\rm sgn}(E(Y_{1}-Y_{0}\mid X))={\rm sgn}(W'\beta)$
with probability one. Hence, 
\[
P(W'b<0)=P(E(Y_{1}-Y_{0}\mid X)<0)=P(W'\beta<0)
\]
and 
\[
P(W'b>0)=P(E(Y_{1}-Y_{0}\mid X)>0)=P(W'\beta>0).
\]

Recall that we have already proved $P(Z'\beta_{1}+\beta_{2}>0)>0$
and $P(Z'\beta_{1}+\beta_{2}<0)>0$, i.e., $P(Z\in A_{1})>0$ and
$P(Z\in B_{1})>0$. Hence, the above two displays imply that both
$P(W'b<0)=P(W'\beta<0)=P(Z\in B_{1})$ and $P(W'b>0)=P(W'\beta>0)=P(Z\in A_{1})$
are strictly positive, i.e., $P(Z\in A_{2})>0$ and $P(Z\in B_{2})>0$. 

Since both $P(Z\in A_{1})$ and $P(Z\in A_{2})$ are strictly positive,
Lemma \ref{lem: topology fact 4} implies that at least one of $\overline{A}_{1}\bigcap\overline{A}_{2}\bigcap\Zcal^{\circ}$
and $\overline{B}_{1}\bigcap\overline{B}_{2}\bigcap\Zcal^{\circ}$
is non-empty. We discuss theses two cases separately.

Suppose that $\overline{A}_{1}\bigcap\overline{A}_{2}\bigcap\Zcal^{\circ}$
is non-empty. By the openness of $A_{1}$ and $A_{2}$, $A_{1}\bigcap A_{2}$
is open. We have seen that $P(Z\in A_{1})$, $P(Z\in A_{2})$, $P(Z\in B_{1})$
and $P(Z\in B_{2})$ are all strictly positive. By Lemma \ref{lem: topology fact 3},
we also have that $\overline{A_{1}\bigcap A_{2}}=\overline{A}_{1}\bigcap\overline{A}_{2}$.
Thus, $\overline{A_{1}\bigcap A_{2}}\bigcap\Zcal^{\circ}=\overline{A}_{1}\bigcap\overline{A}_{2}\bigcap\Zcal^{\circ}\neq\emptyset$.
Hence, we apply Lemma \ref{lem: topology fact 2} with $B=A_{1}\bigcap A_{2}$
and obtain $P(Z\in A_{1}\bigcap A_{2})>0$. By an analogous argument,
we have $P(Z\in B_{1}\bigcap B_{2})>0$ if $\overline{B}_{1}\bigcap\overline{B}_{2}\bigcap\Zcal^{\circ}$
is non-empty. 

Therefore, we have proved that at least one of $P(Z\in A_{1}\bigcap A_{2})$
and $P(Z\in B_{1}\bigcap B_{2})$ is strictly positive. Since $R(b)\geq P(A_{1}\bigcap A_{2})+P(B_{1}\bigcap B_{2})$,
this means that $R(b)>0$. This contradicts $R(b)=0$. Therefore,
we must have that $b=\mu\beta$ for some $\mu>0$. 

\textbf{Step 2:} show that if $b=\mu\beta$ for some $\mu>0$, then
$R(b)=0$.

Since $b=\mu\beta$ with $\mu>0$, ${\rm sgn}(W'b)={\rm sgn}(W'\beta)$
almost surely. By Lemma \ref{lem: manski lem 1}, ${\rm sgn}(E(Y_{1}-Y_{0}\mid X))={\rm sgn}(W'\beta)$
almost surely. Thus, $R(b)=0$.
\end{proof}

\subsection{Proof of Corollary \ref{cor: only cont X}}
\begin{proof}[\textbf{Proof of Corollary \ref{cor: only cont X}}]
Let $d=\dim(X_{t})$. We view this case as a special case of Theorem
\ref{thm: mixed reg} once we treat $\tilde{X}_{t}=(X_{t}',\oneb\{t=0\})'$
as the covariate, where the vector of coefficients is $\tilde{\beta}=(\beta',0)'$.
Since the support of $X_{1}-X_{0}$ contains an open neighborhood
of zero, we can restrict the observations to those such that the support
of $X_{1}-X_{0}$ is an open neighborhood of zero. Moreover, the sign
saturation condition holds because the sign of $E(Y_{1}-Y_{0}\mid X)$
is the same as the sign of $(X_{1}-X_{0})'\beta$. By Theorem \ref{thm: mixed reg},
the vector $(\beta',0)'$ is identified up to scale. Therefore, $\beta$
is identified up to scale. 
\end{proof}

\subsection{Proof of Theorem \ref{thm: discrete cov ID}}
\begin{proof}[\textbf{Proof of Theorem \ref{thm: discrete cov ID}}]
We partition $\beta=(\beta_{(1)}',\beta_{(2)}')'$ with $\beta_{(1)}\in\RR^{K_{1}}$
and $\beta_{(2)}\in\RR^{K_{2}}$. Then $W'\beta=D'\beta_{(1)}+Z'\beta_{(2)}$
and conditional on $D=d_{j}$, we have $W'\beta=d_{j}'\beta_{(1)}+Z'\beta_{(2)}$.
Thus, we are back to the situation in which $W=(1,Z')'$ and the parameter
is $(d_{j}'\beta_{(1)},\beta_{(2)}')'$. By the argument of Theorem
\ref{thm: mixed reg}, we can identify $(d_{j}'\beta_{(1)},\beta_{(2)}')'$
up to scale. Clearly, $\beta_{(2)}\neq0$; otherwise the sign saturation
condition on $D=d_{j}$ would fail. By normalizing $\|\beta_{(2)}\|_{2}=1$,
we identify $d_{j}'\beta_{(1)}$ for each $j\in\{1,...,K_{1}\}$.
Since $(d_{1}'\beta_{(1)},...,d_{K_{1}}'\beta_{(1)})'=(d_{1},...,d_{K_{1}})'\beta_{(1)}$
and $(d_{1},...,d_{K_{1}})'$ has rank $K_{1}$, we identify $\beta_{(1)}$. 
\end{proof}

\subsection{Proof of Theorem \ref{thm: neccessity part 1}}

In the rest of the appendix, we will use the $p(\cdot,\cdot)$ and
$\Acal(\cdot)$ notations introduced in Definition \ref{def: chamb alter}.
Recall that for any $t\in\RR$, $\Acal(t)={\rm conv}\{p(t,\alpha):\alpha\in\RR\}$
and 
\[
p(t,\alpha)=\begin{pmatrix}F(\alpha)\\
F(t+\alpha)\\
F(\alpha)\cdot F(t+\alpha)
\end{pmatrix}.
\]

We will also repeatedly use the notation $G(a)=\ln\frac{F(a)}{1-F(a)}$
and $\dot{G}(t)=dG(t)/dt$. Throughout the rest of the appendix, we
assume that $F(\cdot)$ is strictly increasing on $\RR$, which means
that $G(\cdot)$ is also strictly increasing on $\RR$. 
\begin{lem}
\label{lem: farkas} Let $s,t\in\RR$. Suppose that $\sup_{\alpha\in\RR}v'p(s,\alpha)>\inf_{\alpha\in\RR}v'p(t,\alpha)$
for any $v\neq(0,0,0)'$. Then $\Acal(s)\bigcap\Acal(t)\neq\emptyset$. 
\end{lem}
\begin{proof}
We proceed by contradiction. Suppose that $\Acal(s)\bigcap\Acal(t)=\emptyset$.
Notice that both $\Acal(s)$ and $\Acal(t)$ are convex and non-empty.
By the separating hyperplane theorem (e.g., Theorem 7.30 of \citet{aliprantis2006infinite}),
there exists $v_{0}\neq(0,0,0)'$ and $r\in\RR$ such that 
\[
v_{0}'a\leq r\qquad\forall a\in\Acal(s)
\]
and 
\[
v_{0}'b\geq r\qquad\forall b\in\Acal(t).
\]

In other words, 
\[
\int v_{0}'p(s,\alpha)d\pi_{1}(\alpha)\leq r\qquad\forall\pi_{1}\in\Pi
\]
and 
\[
\int v_{0}'p(t,\alpha)d\pi_{2}(\alpha)\geq r\qquad\forall\pi_{2}\in\Pi.
\]

By the assumption and $v_{0}\neq0$, we have that $\sup_{\alpha}v_{0}'p(s,\alpha)>\inf_{\alpha}v_{0}'p(t,\alpha)$.
Thus, there exists $\varepsilon>0$ such that 
\begin{equation}
\sup_{\alpha}v_{0}'p(s,\alpha)>\varepsilon+\inf_{\alpha}v_{0}'p(t,\alpha).\label{eq: lem farkas eq 3}
\end{equation}

We now choose $\alpha_{*}$ and $\xi_{*}$ such that $v_{0}'p(s,\alpha_{*})\geq\sup_{\alpha}v_{0}'p(s,\alpha)-\varepsilon/4$
and $v_{0}'p(t,\xi_{*})\leq\inf_{\xi}v_{0}'p(t,\xi)+\varepsilon/4$.
Choose $\pi_{1,*}$ to be the probability measure that puts all the
mass on $\alpha_{*}$. Similarly, choose $\pi_{2,*}$ to be the probability
measure that puts all the mass on $\xi_{*}$. Then $v_{0}'p(s,\alpha_{*})=\int v_{0}'p(s,\alpha)d\pi_{1,*}(\alpha)\leq r$
and $v_{0}'p(t,\xi_{*})=\int v_{0}'p(t,\alpha)d\pi_{2,*}(\alpha)\geq r$.
Hence, 
\begin{equation}
v_{0}'p(t,\xi_{*})\geq v_{0}'p(s,\alpha_{*}).\label{eq: lem farkas eq 4}
\end{equation}

By (\ref{eq: lem farkas eq 3}) and the definitions of $\alpha_{*}$
and $\xi_{*}$, we have 
\[
v_{0}'p(s,\alpha_{*})+\varepsilon/4\geq\sup_{\alpha}v_{0}'p(s,\alpha)>\varepsilon+\inf_{\alpha}v_{0}'p(t,\alpha)\geq\varepsilon+v_{0}'p(t,\xi_{*})-\varepsilon/4.
\]

This means that $v_{0}'p(s,\alpha_{*})\geq\varepsilon/2+v_{0}'p(t,\xi_{*})$.
Since $\varepsilon>0$, this contradicts (\ref{eq: lem farkas eq 4}).
\end{proof}
\begin{lem}
\label{lem: key ID}Suppose that $s>t$. Then for any $v\in\RR^{3}$,
$\sup_{\alpha\in\RR}v'p(s,\alpha)\geq\inf_{\alpha\in\RR}v'p(t,\alpha)$. 
\end{lem}
\begin{proof}
Suppose not. Then there exist $v\in\RR^{3}$ and $\varepsilon>0$
such that 
\[
\varepsilon+\sup_{\alpha\in\RR}v'p(s,\alpha)<\inf_{\xi\in\RR}v'p(t,\xi).
\]

Let $\delta=s-t$. By assumption, $\delta>0$. The above display implies
\[
v'[p(t+\delta,\alpha)-p(t,\xi)]\leq-\varepsilon\qquad\forall(\alpha,\xi)\in\RR^{2}.
\]

Let $v=(v_{1},v_{2},v_{3})'$. Then for any $\alpha,\xi\in\RR$, 
\begin{equation}
v_{1}[F(\alpha)-F(\xi)]+v_{2}[F(t+\delta+\alpha)-F(t+\xi)]+v_{3}[F(\alpha)F(t+\delta+\alpha)-F(\xi)F(t+\xi)]\leq-\varepsilon.\label{eq: lem key ID eq 4}
\end{equation}

Taking $\xi=\alpha$, we obtain that for any $\alpha\in\RR$, 
\[
v_{2}[F(t+\delta+\alpha)-F(t+\alpha)]+v_{3}F(\alpha)[F(t+\delta+\alpha)-F(t+\alpha)]\leq-\varepsilon.
\]

This means that for any $\alpha\in\RR$, 
\[
v_{2}\leq-v_{3}F(\alpha)-\frac{\varepsilon}{F(t+\delta+\alpha)-F(t+\alpha)}.
\]

Notice that 
\begin{multline*}
\inf_{\alpha\in\RR}\left[-v_{3}F(\alpha)-\frac{\varepsilon}{F(t+\delta+\alpha)-F(t+\alpha)}\right]\leq\inf_{\alpha\in\RR}\left[|v_{3}|-\frac{\varepsilon}{F(t+\delta+\alpha)-F(t+\alpha)}\right]\\
\leq\lim_{\alpha\rightarrow-\infty}\left[|v_{3}|-\frac{\varepsilon}{F(t+\delta+\alpha)-F(t+\alpha)}\right]=-\infty.
\end{multline*}

Thus, $v_{2}\leq-\infty$. This contradicts $v\in\RR^{3}$.
\end{proof}
\begin{lem}
\label{lem: ID part 1}Assume that $s>t>0$. If $\sup_{\alpha\in\RR}v'p(s,\alpha)=\inf_{\alpha\in\RR}v'p(t,\alpha)$
for some $v\neq(0,0,0)'$, then 
\[
\inf_{\alpha\in\RR}[G(s+\alpha)-G(\alpha)]\geq\sup_{\xi\in\RR}[G(t+\xi)-G(\xi)].
\]
\end{lem}
\begin{proof}
Let $v=(v_{1},v_{2},v_{3})'$. Without loss of generality, we can
normalize and assume $v_{3}\in\{0,1,-1\}$. Notice that $\lim_{\alpha\rightarrow\infty}v'p(t,\alpha)=0$
and $\lim_{\alpha\rightarrow\infty}v'p(s,\alpha)=0$. Thus, $\inf_{\alpha\in\RR}v'p(t,\alpha)\leq0\leq\sup_{\alpha\in\RR}v'p(s,\alpha)$.
By $\sup_{\alpha\in\RR}v'p(s,\alpha)=\inf_{\alpha\in\RR}v'p(t,\alpha)$,
it follows that 
\[
\sup_{\alpha\in\RR}v'p(s,\alpha)=\inf_{\alpha\in\RR}v'p(t,\alpha)=0.
\]

Let $\delta=s-t$. By assumption, $\delta>0$ and $t>0$. Then for
any $\alpha,\xi\in\RR$, $v'p(s,\alpha)\leq0\leq v'p(t,\xi)$. In
other words, for any $\alpha,\xi\in\RR$,
\begin{multline}
v_{1}F(\alpha)+v_{2}F(t+\delta+\alpha)+v_{3}F(\alpha)F(t+\delta+\alpha)\leq0\\
\leq v_{1}F(\xi)+v_{2}F(t+\xi)+v_{3}F(\xi)F(t+\xi).\label{eq: lem ID part 1 eq 3}
\end{multline}

We now proceed in three steps. 

\textbf{Step 1:} rule out $v_{3}=0$.

Suppose that $v_{3}=0$. By (\ref{eq: lem ID part 1 eq 3}), we have
that for any $\alpha,\xi\in\RR$,

\begin{equation}
v_{1}[F(\alpha)-F(\xi)]+v_{2}[F(t+\delta+\alpha)-F(t+\xi)]\leq0.\label{eq: lem ID part 1 eq 4}
\end{equation}

Take $\alpha=\xi-\delta$ and obtain $v_{1}[F(\xi-\delta)-F(\xi)]\leq0$
for any $\xi\in\RR$. By $\delta>0$, this means that $v_{1}\geq0$. 

It is easy to verify $v_{1}+v_{2}=0$: taking $\alpha=-\infty$ in
(\ref{eq: lem ID part 1 eq 4}), we have that $-v_{1}F(\xi)-v_{2}F(t+\xi)\leq0$
for any $\xi\in\RR$, which means that $v_{1}+v_{2}\geq0$ (take $\xi\rightarrow\infty$);
taking $\alpha=\infty$ in (\ref{eq: lem ID part 1 eq 4}), we have
that $v_{1}(1-F(\xi))+v_{2}(1-F(t+\xi))\leq0$ for any $\xi\in\RR$,
which means that $v_{1}+v_{2}\leq0$ (take $\xi\rightarrow-\infty$). 

Hence, $v=(v_{1},-v_{1},0)$. By $v\neq(0,0,0)'$ and $v_{1}\geq0$,
we have $v_{1}>0$. By (\ref{eq: lem ID part 1 eq 4}), we have that
for any $\alpha,\xi\in\RR$,
\[
v_{1}[F(\alpha)-F(\xi)]-v_{1}[F(t+\delta+\alpha)-F(t+\xi)]\leq0.
\]

Since $v_{1}>0$, it follows that for any $\alpha,\xi\in\RR$,
\[
F(\alpha)-F(\xi)\leq F(t+\delta+\alpha)-F(t+\xi).
\]

Taking $\alpha=\infty$, we have $F(\xi)\geq F(t+\xi)$, which is
impossible because $t>0$. Thus, $v_{3}\neq0$. 

\textbf{Step 2:} rule out $v_{3}=1$.

Suppose that $v_{3}=1$. Then (\ref{eq: lem ID part 1 eq 3}) implies
that for any $\alpha,\xi\in\RR$, 
\begin{equation}
v_{1}[F(\alpha)-F(\xi)]+v_{2}[F(t+\delta+\alpha)-F(t+\xi)]+F(\alpha)F(t+\delta+\alpha)-F(\xi)F(t+\xi)\leq0.\label{eq: lem ID part 1 eq 5}
\end{equation}

Taking $\alpha\rightarrow\infty$ and $\xi\rightarrow-\infty$, we
have $v_{1}+v_{2}+1\leq0$. Taking $\alpha\rightarrow-\infty$ and
$\xi\rightarrow\infty$, we have $-v_{1}-v_{2}-1\leq0$, which means
$v_{1}+v_{2}+1\geq0$. Hence, $v_{1}+v_{2}+1=0$.

Take $\alpha=\xi-\delta$ in (\ref{eq: lem ID part 1 eq 5}) and obtain
that for any $\xi\in\RR$,
\[
v_{1}[F(\xi-\delta)-F(\xi)]+[F(\xi-\delta)-F(\xi)]F(t+\xi)\leq0.
\]

Since $\delta>0$, we have $v_{1}+F(t+\xi)\geq0$ for any $\xi\in\RR$.
Taking $\xi\rightarrow-\infty$, we have $v_{1}\geq0$. 

By $v_{3}=1$ and $v_{1}+v_{2}+1=0$, the right-hand side of (\ref{eq: lem ID part 1 eq 3})
implies that for any $\xi\in\RR$, 
\[
v_{1}F(\xi)+(-1-v_{1})F(t+\xi)+F(\xi)F(t+\xi)\geq0.
\]

Taking $\xi=0$, we have $[F(0)-1]F(t)\geq v_{1}[F(t)-F(0)]$. Since
$t>0$, we have 
\[
v_{1}\leq\frac{[F(0)-1]F(t)}{F(t)-F(0)}<0.
\]

This contradicts $v_{1}\geq0$. Thus, $v_{3}\neq1$. 

\textbf{Step 3:} show the final result.

From the previous two steps, we ruled out $v_{3}\in\{0,1\}$. Hence,
$v_{3}=-1$. Therefore, (\ref{eq: lem ID part 1 eq 3}) implies that
 for any $\alpha,\xi\in\RR$, 
\begin{equation}
[v_{1}-F(t+\delta+\alpha)][F(\alpha)-F(\xi)]+[v_{2}-F(\xi)][F(t+\delta+\alpha)-F(t+\xi)]\leq0\label{eq: lem ID part 1 eq 6}
\end{equation}

Taking $\alpha=\xi-\delta$ gives $v_{1}-F(t+\xi)\geq0$ (due to $\delta>0$)
for any $\xi\in\RR$, which means that $v_{1}\geq1$. 

By the right-hand side of (\ref{eq: lem ID part 1 eq 3}) and $v_{3}=-1$,
$F(\xi)v_{1}+v_{2}F(t+\xi)-F(\xi)F(t+\xi)\geq0$ for any $\xi\in\RR$.
Taking $\xi=\infty$ gives $v_{1}+v_{2}\geq1$. On the other hand,
the left-hand side of (\ref{eq: lem ID part 1 eq 3}) and $v_{3}=-1$
yield $v_{1}F(\alpha)+v_{2}F(t+\delta+\alpha)-F(\alpha)F(t+\delta+\alpha)\leq0$
for any $\alpha\in\RR$, which (by taking $\alpha\rightarrow\infty$)
yields $v_{1}+v_{2}\leq1$. Hence, $v_{1}+v_{2}=1$, which means that
$v_{2}=1-v_{1}$ and $v=(v_{1},1-v_{1},-1)'$. 

Using the left-hand side of (\ref{eq: lem ID part 1 eq 3}) and $v=(v_{1},1-v_{1},-1)'$,
we have that for any $\alpha\in\RR$, 
\[
v_{1}F(\alpha)+(1-v_{1})F(t+\delta+\alpha)-F(\alpha)F(t+\delta+\alpha)\leq0.
\]

Since $F(t+\delta+\alpha)>F(\alpha)$ (due to $t+\delta>0$), we have
that for any $\alpha\in\RR$,
\[
v_{1}\geq\frac{F(t+\delta+\alpha)[1-F(\alpha)]}{F(t+\delta+\alpha)-F(\alpha)}.
\]

Hence, 
\begin{equation}
v_{1}\geq\sup_{\alpha\in\RR}\frac{F(t+\delta+\alpha)[1-F(\alpha)]}{F(t+\delta+\alpha)-F(\alpha)}.\label{eq: lem ID part 1 eq 7}
\end{equation}

Similarly, the right-hand side of (\ref{eq: lem ID part 1 eq 3})
and $v=(v_{1},1-v_{1},-1)'$, we have that for any $\xi\in\RR$, $v_{1}F(\xi)+(1-v_{1})F(t+\xi)-F(\xi)F(t+\xi)\geq0$.
Hence, 
\begin{equation}
v_{1}\leq\inf_{\xi\in\RR}\frac{F(t+\xi)[1-F(\xi)]}{F(t+\xi)-F(\xi)}.\label{eq: lem ID part 1 eq 8}
\end{equation}

Combining (\ref{eq: lem ID part 1 eq 7}) and (\ref{eq: lem ID part 1 eq 8}),
we have
\begin{equation}
\sup_{\alpha\in\RR}\frac{F(t+\delta+\alpha)[1-F(\alpha)]}{F(t+\delta+\alpha)-F(\alpha)}\leq\inf_{\xi\in\RR}\frac{F(t+\xi)[1-F(\xi)]}{F(t+\xi)-F(\xi)}.\label{eq: lem ID part 1 eq 9}
\end{equation}

Recall $G(a)=\ln\frac{F(a)}{1-F(a)}$. Then with straight-forward
algebra, (\ref{eq: lem ID part 1 eq 9}) becomes $\inf_{\alpha\in\RR}\left[G(t+\delta+\alpha)-G(\alpha)\right]\geq\sup_{\xi\in\RR}\left[G(t+\xi)-G(\xi)\right]$.
\end{proof}
\begin{lem}
\label{lem: ID part 2} Let $T\subset(0,\infty)$ be a compact set.
Assume that $G(\cdot)$ is continuously differentiable. Suppose that
for any $\delta>0$, there exists $t\in T$ such that 
\[
\inf_{\alpha\in\RR}[G(t+\delta+\alpha)-G(\alpha)]\geq\sup_{\xi\in\RR}[G(t+\xi)-G(\xi)].
\]

Then $\dot{G}(\cdot)$ is a periodic function with a period in $T$.
In other words, there exists $t\in T$ such that $\dot{G}(t+\xi)=\dot{G}(\xi)$
for any $\xi\in\RR$. 
\end{lem}
\begin{proof}
For any $b\in\RR$, define $r_{b}(\cdot)$ by $r_{b}(a)=G(b+a)-G(a)$.
We also define $c(b)=\sup_{\xi\in\RR}r_{b}(\xi)$. Let $K>0$ be arbitrary.
Then by assumption, for any $\delta>0$, there exists $t\in T$ (depending
only on $\delta$) such that $\inf_{\alpha\in\RR}[r_{t}(\delta+\alpha)+r_{\delta}(\alpha)]\geq\sup_{\xi\in\RR}r_{t}(\xi)=c(t)$,
which means that $\inf_{\alpha+\delta\in[-K,K]}[r_{t}(\delta+\alpha)+r_{\delta}(\alpha)]\geq c(t)$
and thus
\[
c(t)\leq\inf_{|\delta+\alpha|\leq K}r_{t}(\delta+\alpha)+\sup_{|\delta+\alpha|\leq K}r_{\delta}(\alpha)=\inf_{|\xi|\leq K}r_{t}(\xi)+\sup_{|\delta+\alpha|\leq K}r_{\delta}(\alpha).
\]

Notice that $\sup_{|\delta+\alpha|\leq K}r_{\delta}(\alpha)=\sup_{|\delta+\alpha|\leq K}[G(\delta+\alpha)-G(\alpha)]=\sup_{|\xi|\leq K}[G(\xi)-G(\xi-\delta)]$.
We have that for any $\delta,K>0$, there exists $t\in T$ depending
only on $\delta$ such that 
\[
\sup_{|\xi|\leq K}[G(\xi)-G(\xi-\delta)]\geq c(t)-\inf_{|\xi|\leq K}r_{t}(\xi)=\sup_{\xi\in\RR}r_{t}(\xi)-\inf_{|\xi|\leq K}r_{t}(\xi)\geq0.
\]

Now we choose $\delta=1/n$. Then for any $n>1$, there exists $t_{n}$
(not depending on $K$) such that 
\[
\sup_{|\xi|\leq K}[G(\xi)-G(\xi-1/n)]\geq c(t_{n})-\inf_{|\xi|\leq K}r_{t_{n}}(\xi)\geq0.
\]

Let $q_{n}(K):=c(t_{n})-\inf_{|\xi|\leq K}r_{t_{n}}(\xi)$. Since
$G(\cdot)$ is continuous, it is uniformly continuous on compact sets
(Heine-Cantor theorem). Hence, $\limsup_{n\rightarrow\infty}\sup_{|\xi|\leq K}[G(\xi)-G(\xi-1/n)]=0$.
The above display implies $\limsup_{n\rightarrow\infty}q_{n}(K)=0$. 

Since $t_{n}$ is in a compact set $T$, there exists a subsequence
$t_{n_{j}}$ and $t_{*}\in T$ such that $t_{n_{j}}\rightarrow t_{*}$
and $t_{*}$ does not depend on $K$. We can further extract a subsequence
that is either increasing or decreasing. With an abuse of notation,
we write $t_{n}$ rather than $t_{n_{j}}$ in the rest of the proof.
We show the following claim in two cases ($t_{n}\uparrow t_{*}$ and
$t_{n}\downarrow t_{*}$): 
\begin{equation}
\sup_{|\xi|\leq K}r_{t_{*}}(\xi)=\inf_{|\xi|\leq K}r_{t_{*}}(\xi).\label{eq: lem ID part 2 eq 2}
\end{equation}

\textbf{Step 1: }show (\ref{eq: lem ID part 2 eq 2}) assuming $t_{n}\uparrow t_{*}$. 

Since $t_{n}\leq t_{*}$, we have $r_{t_{n}}(\xi)\leq r_{t_{*}}(\xi)$
for any $\xi$ and thus
\[
c(t_{n})=q_{n}+\inf_{|\xi|\leq K}r_{t_{n}}(\xi)\leq q_{n}+\inf_{|\xi|\leq K}r_{t_{*}}(\xi).
\]

Hence, by $\limsup_{n\rightarrow\infty}q_{n}(K)=0$, we have
\begin{equation}
\limsup_{n\rightarrow\infty}c(t_{n})\leq\inf_{|\xi|\leq K}r_{t_{*}}(\xi).\label{eq: lem ID part 2 eq 3}
\end{equation}

Let $R=\max_{t\in T}|t|$. Since $T$ is compact, $R$ is bounded.
Notice that $G(\cdot)$ is continuous. Again, by the Heine-Cantor
theorem, it is uniformly continuous on the compact set $[-K-R,K+R]$.
Notice that $t_{n}+\xi\in[-K-R,K+R]$ for any $\xi\in[-K,K]$ (due
to $|t_{n}|\leq R$). Thus, $\sup_{|\xi|\leq K}|r_{t_{n}}(\xi)-r_{t_{*}}(\xi)|=\sup_{|\xi|\leq K}|G(t_{n}+\xi)-G(t_{*}+\xi)|\rightarrow0.$
It follows that 
\[
\limsup_{n\rightarrow\infty}\sup_{|\xi|\leq K}r_{t_{n}}(\xi)=\sup_{|\xi|\leq K}r_{t_{*}}(\xi).
\]

Now by (\ref{eq: lem ID part 2 eq 3}) and the above display, we have
that 
\[
\inf_{|\xi|\leq K}r_{t_{*}}(\xi)\geq\limsup_{n\rightarrow\infty}c(t_{n})=\limsup_{n\rightarrow\infty}\sup_{\xi\in\RR}r_{t_{n}}(\xi)\geq\limsup_{n\rightarrow\infty}\sup_{|\xi|\leq K}r_{t_{n}}(\xi)=\sup_{|\xi|\leq K}r_{t_{*}}(\xi).
\]

On the other hand, we have $\inf_{|\xi|\leq K}r_{t_{*}}(\xi)\leq\sup_{|\xi|\leq K}r_{t_{*}}(\xi)$.
This proves (\ref{eq: lem ID part 2 eq 2}).

\textbf{Step 2: }show (\ref{eq: lem ID part 2 eq 2}) assuming $t_{n}\downarrow t_{*}$. 

Since the argument is similar to Step 1, we only provide an outline
here. By $t_{n}\geq t_{*}$, we have 
\[
q_{n}+\inf_{|\xi|\leq K}r_{t_{n}}(\xi)=c(t_{n})\geq c(t_{*})=\sup_{\xi\in\RR}r_{t_{*}}(\xi)\geq\sup_{|\xi|\leq K}r_{t_{*}}(\xi).
\]

Notice that $r_{t_{*}}(\xi)=G(t_{*}+\xi)-G(\xi)$ and $r_{t_{n}}(\xi)=G(t_{n}+\xi)-G(\xi)$.
Again, by the Heine-Cantor theorem and the continuity of $G(\cdot)$,
we have $\sup_{|\xi|\leq K}|r_{t_{n}}(\xi)-r_{t_{*}}(\xi)|=\sup_{|\xi|\leq K}|G(t_{n}+\xi)-G(t_{*}+\xi)|\rightarrow0$
due to $t_{n}\rightarrow t_{*}$. Thus, 
\[
\liminf_{n\rightarrow\infty}\inf_{|\xi|\leq K}r_{t_{n}}(\xi)=\inf_{|\xi|\leq K}r_{t_{*}}(\xi).
\]

The above two displays and $q_{n}\rightarrow0$ imply $\inf_{|\xi|\leq K}r_{t_{*}}(\xi)\geq\sup_{|\xi|\leq K}r_{t_{*}}(\xi)$.
This gives (\ref{eq: lem ID part 2 eq 2}).

\textbf{Step 3: }show the final result.

Notice that we have verified (\ref{eq: lem ID part 2 eq 2}) for an
arbitrary $K>0$. Thus, $r_{t_{*}}(\cdot)$ is a constant function
on $[-K,K]$ for any $K>0$. Since $t_{*}$ does not depend on $K$,
$r_{t_{*}}(\cdot)$ is a constant function on $\RR$. This means that
the derivative of $r_{t_{*}}(\cdot)$ is zero on $\RR$. Recalling
that $r_{t_{*}}(\xi)=G(t_{*}+\xi)-G(\xi)$, we have $\dot{G}(t_{*}+\xi)=\dot{G}(\xi)$
for any $\xi\in\RR$.
\end{proof}
\begin{proof}[\textbf{Proof of Theorem \ref{thm: neccessity part 1}}]
Fix an arbitrary $\beta=(\beta_{1}',\beta_{2})'\in\Bcal_{+}$. Then
$w'\beta=z'\beta_{1}+\beta_{2}>0$ for any $z\in\Zcal$. Let $\delta>0$
be a number to be determined. Define $b=(\beta_{1}',\beta_{2}+\delta)'$.
Clearly, $w'b=z'\beta_{1}+\beta_{2}+\delta>0$ for any $z\in\Zcal$.
Thus, $b\in\Bcal_{+}$. Define $T$ to be the closure of $\{z'\beta_{1}+\beta_{2}:\ z\in\Zcal\}$.
Since $\Zcal$ is bounded, $T$ is also bound. By the closedness of
$T$, $T$ is also compact. We also notice that $0\notin T$ and thus
$T\subset(0,\infty)$. (To see this, notice that $\inf_{z\in\Zcal}(z',1)'\beta$
is achieved by some $z_{*}\in\Zcal$ (because $\Zcal$ is closed and
bounded). On the other hand, since for any $z\in\Zcal$, $(z',1)'\beta>0$,
we have that $(z_{*}',1)'\beta>0$.)

We now show that there exists $\delta>0$ such that for any $z\in\Zcal$,
$\Acal(w'\beta)\bigcap\Acal(w'b)\neq\emptyset$, where $w=(z',1)'$.
We proceed by contradiction. Suppose that for any $\delta>0$, there
exists $z_{\delta}\in\Zcal$ such that $\Acal(w_{\delta}'b)\bigcap\Acal(w_{\delta}'\beta)=\emptyset$,
where $w_{\delta}=(z_{\delta}',1)'$. By Lemma \ref{lem: farkas},
there exists $v_{\delta}\neq0$ such that $\sup_{\alpha\in\RR}v_{\delta}'p(w_{\delta}'b,\alpha)\leq\inf_{\alpha\in\RR}v_{\delta}'p(w_{\delta}'\beta,\alpha)$.
By Lemma \ref{lem: key ID} and $w_{\delta}'b-w_{\delta}'\beta=\delta>0$,
we have that $\sup_{\alpha\in\RR}v_{\delta}'p(w_{\delta}'b,\alpha)\geq\inf_{\alpha\in\RR}v_{\delta}'p(w_{\delta}'\beta,\alpha)$.
Hence, 
\[
\sup_{\alpha\in\RR}v_{\delta}'p(w_{\delta}'\beta+\delta,\alpha)=\sup_{\alpha\in\RR}v_{\delta}'p(w_{\delta}'b,\alpha)=\inf_{\alpha\in\RR}v_{\delta}'p(w_{\delta}'\beta,\alpha).
\]

By Lemma \ref{lem: ID part 1} (with $t=w_{\delta}'\beta$),
\[
\inf_{\alpha\in\RR}[G(w_{\delta}'\beta+\delta+\alpha)-G(\alpha)]\geq\sup_{\xi\in\RR}[G(w_{\delta}'\beta+\xi)-G(\xi)].
\]

Notice that $w_{\delta}'\beta\in T$. Therefore, we have shown that
for any $\delta>0$, there exists $t_{\delta}\in T$ such that 
\[
\inf_{\alpha\in\RR}[G(t_{\delta}+\delta+\alpha)-G(\alpha)]\geq\sup_{\xi\in\RR}[G(t_{\delta}+\xi)-G(\xi)].
\]

By Lemma \ref{lem: ID part 2}, $\dot{G}$ is a periodic function.
However, this contradicts the assumption that $\dot{G}$ is not a
periodic function. 
\end{proof}

\subsection{Proof of Theorem \ref{thm: ID period fun}}
\begin{lem}
\label{lem: ID part 3}Suppose that $h(\cdot)$ is a continuous and
non-constant function on $\RR$. If $h(\cdot)$ is a periodic function,
then $h(\cdot)$ has a minimal positive period, i.e., the set $\{a>0:\ h(a+x)=h(x)\ \forall x\in\RR\}$
has a smallest element. 
\end{lem}
\begin{proof}
Let $K=\{a>0:\ h(a+x)=h(x)\ \forall x\in\RR\}$. 

We first show that $\inf K>0$. Suppose that this is not true. Then
$\inf K=0$. Thus, there exists $a_{n}\in K$ such that $a_{n}\downarrow0$.
Let $x,y\in\RR$ such that $x<y$. For any $n$, define $J_{n}$ to
be the integer part of $(y-x)/a_{n}$, i.e., $J_{n}$ is the integer
satisfying $J_{n}a_{n}\leq y-x<(J_{n}+1)a_{n}$. Since $a_{n}$ is
a period for $h(\cdot)$ and $J_{n}$ is an integer, $J_{n}a_{n}$
is also a period of $h(\cdot)$. Thus, $h(y)=h(y-x+x)=h(y-x-J_{n}a_{n}+x)$.
Notice that $0\leq y-x-J_{n}a_{n}<a_{n}\rightarrow0$. Thus, $(y-x-J_{n}a_{n})+x\rightarrow x$.
By the continuity of $h(\cdot)$, we have that $h(y-x-J_{n}a_{n}+x)\rightarrow h(x)$.
Hence, we have proved that $h(y)=h(x)$. Since $x,y$ are two arbitrary
numbers with $x<y$, this means that $h(\cdot)$ is a constant function
on $\RR$. This would contradict the assumption that $h(\cdot)$ is
non-constant. Hence, $\inf K>0$. 

Let $\eta=\inf K$. Since $\eta>0$ and $\eta\leq a$ for any $a\in K$,
we only need to show that $\eta\in K$. By definition, there exists
$a_{n}\in K$ with $a_{n}\rightarrow\eta$. For any $x$, we apply
the continuity of $h(\cdot)$ and obtain that $h(x)=h(x+a_{n})\rightarrow h(x+\eta)$,
which means that $h(x)=h(x+\eta)$ and thus $\eta\in K$. Therefore,
$\eta$ is the smallest element of $K$.
\end{proof}
\begin{lem}
\label{lem: ID part 4}Define the function $G(a)=\ln\frac{F(a)}{1-F(a)}$.
Assume that $s>t>0$. If 
\[
\inf_{\alpha\in\RR}[G(s+\alpha)-G(\alpha)]>\sup_{\xi\in\RR}[G(t+\xi)-G(\xi)]>0,
\]
then $\Acal(s)\bigcap\Acal(t)=\emptyset$. 
\end{lem}
\begin{proof}
Let $\delta=s-t$. By $s>t$, $\delta>0$. We proceed in two steps. 

\textbf{Step 1:} show that there exists $v\in\RR^{3}$ such that $v'p(s,\alpha)<0<v'p(t,\xi)$
for any $\alpha,\xi\in\RR$.

Consider $v=(v_{1},1-v_{1},-1)'$. We would like to choose $v_{1}$
such that for any $\alpha,\xi\in\RR$
\begin{equation}
v_{1}F(\alpha)+(1-v_{1})F(s+\alpha)-F(\alpha)F(s+\alpha)<0\label{eq: lem ID part 4 eq 3}
\end{equation}
and 
\begin{equation}
v_{1}F(\xi)+(1-v_{1})F(t+\xi)-F(\xi)F(t+\xi)>0.\label{eq: lem ID part 4 eq 4}
\end{equation}

Since $s>t>0$, it suffices to choose any $v_{1}$ such that 
\[
\sup_{\alpha\in\RR}\frac{F(s+\alpha)[1-F(\alpha)]}{F(s+\alpha)-F(\alpha)}<v_{1}<\inf_{\xi\in\RR}\frac{F(t+\xi)[1-F(\xi)]}{F(t+\xi)-F(\xi)}.
\]

By $F(\cdot)=[\exp(-G(\cdot))+1]^{-1}$, this means that we can choose
any $v_{1}$ such that 
\[
\frac{1}{1-\sup_{\alpha\in\RR}\exp[G(\alpha)-G(s+\alpha)]}<v_{1}<\frac{1}{1-\inf_{\xi\in\RR}\exp[G(\xi)-G(t+\xi)]}.
\]

Since $\inf_{\alpha\in\RR}[G(s+\alpha)-G(\alpha)]>\sup_{\xi\in\RR}[G(t+\xi)-G(\xi)]>0$,
we have $\sup_{\alpha\in\RR}[G(\alpha)-G(s+\alpha)]<\inf_{\xi\in\RR}[G(\xi)-G(t+\xi)]<0$
and thus a choice of $v_{1}$ in the above display is clearly possible.
Therefore, by $v=(v_{1},1-v_{1},-1)'$, we have found $v\in\RR^{3}$
such that (\ref{eq: lem ID part 4 eq 3}) and (\ref{eq: lem ID part 4 eq 4})
hold. 

\textbf{Step 2:} show the final result.

We proceed by contradiction. Suppose that $\Acal(s)\bigcap\Acal(t)\neq\emptyset$.
Then there exists $r\in\Acal(s)\bigcap\Acal(t)$. Since $r\in\Acal(s)\subset\RR^{3}$,
by Caratheodory's theorem (e.g., Theorem 5.32 in \citet{aliprantis2006infinite}),
there exist $\alpha_{1},\alpha_{2},\alpha_{3},\alpha_{4}\in\RR$ and
$\lambda_{1},...,\lambda_{4}\geq0$ such that $r=\sum_{j=1}^{4}\lambda_{j}p(s,\alpha_{j})$
and $\sum_{j=1}^{4}\lambda_{j}=1$. Similarly, $r=\sum_{j=1}^{4}\rho_{j}p(t,\xi_{j})$,
where $\xi_{1},...,\xi_{4}\in\RR$ and $\rho_{1},...,\rho_{4}\geq0$
satisfy $\sum_{j=1}^{4}\rho_{j}=1$. Therefore, for any $v\in\RR^{3}$,
\[
\max_{1\leq j\leq4}p(s,\alpha_{j})'v\geq\sum_{j=1}^{4}\lambda_{j}p(s,\alpha_{j})'v=r'v=\sum_{j=1}^{4}\rho_{j}p(t,\xi_{j})'v\geq\min_{1\leq j\leq4}p(t,\xi_{j})'v.
\]

However, in Step 1, we have shown that there exists $v\neq(0,0,0)'$
such that $v'p(s,\alpha)<0<v'p(t,\xi)$ for any $\alpha,\xi\in\RR$.
This is a contradiction. 
\end{proof}
\begin{lem}
\label{lem: ID part 5} Assume that $s>t>0$. Suppose that $\dot{G}$
is a periodic function with a positive period $\eta>0$. If either
$s/\eta$ or $t/\eta$ is an integer, then there exists $\varepsilon\in(0,(s-t)/4]$
such that 
\[
\inf_{\alpha\in\RR}[G(s-\varepsilon+\alpha)-G(\alpha)]>\sup_{\xi\in\RR}[G(t+\varepsilon+\xi)-G(\xi)]>0.
\]
\end{lem}
\begin{proof}
Since $\dot{G}(l\cdot\eta+\xi)-\dot{G}(\xi)=0$ for any $\xi\in\RR$
and for any $l\in\ZZ$ ($\ZZ$ denotes the set of all integers), it
follows that $G(l\cdot\eta+\xi)-G(\xi)$ does not depend on $\xi$,
which means that $G(l\cdot\eta+\xi)-G(\xi)=G(l\cdot\eta+0)-G(0)$.
We notice that
\[
G(l\cdot\eta)=G(0)+\int_{0}^{l\cdot\eta}\dot{G}(t)dt=G(0)+\sum_{j=1}^{l}\int_{\eta(j-1)}^{\eta j}\dot{G}(t)dt=G(0)+\sum_{j=1}^{l}\int_{0}^{\eta}\dot{G}(t)dt=lq_{0}+G(0),
\]
where $q_{0}=G(\eta)-G(0)$. Thus, we have that for any $\xi\in\RR$
and for any $l\in\ZZ$,
\begin{equation}
G(l\cdot\eta+\xi)=G(\xi)+lq_{0}.\label{eq: lem ID part 5 eq 3}
\end{equation}

Notice that we can represent all the real numbers as $(k+r)\eta$
with $k\in\ZZ$ and $r\in[0,1)$. Then it suffices to show that for
some $\varepsilon>0$, 
\begin{multline*}
\inf_{(k_{1},r_{1})\in\ZZ\times[0,1)}[G(s-\varepsilon+(k_{1}+r_{1})\eta)-G((k_{1}+r_{1})\eta)]\\
>\sup_{(k_{2},r_{2})\in\ZZ\times[0,1)}[G(t+\varepsilon+(k_{2}+r_{2})\eta)-G((k_{2}+r_{2})\eta)]>0.
\end{multline*}

By (\ref{eq: lem ID part 5 eq 3}), $G(s-\varepsilon+(k_{1}+r_{1})\eta)-G((k_{1}+r_{1})\eta)=G(s-\varepsilon+r_{1}\eta)-G(r_{1}\eta)$,
which does not depend on $k_{1}$; similarly, $G(t+\varepsilon+(k_{2}+r_{2})\eta)-G((k_{2}+r_{2})\eta)=G(t+\varepsilon+r_{2}\eta)-G(r_{2}\eta)$.
Thus, it suffices to show that for some $\varepsilon\in(0,(s-t)/4]$,
\begin{equation}
\inf_{r_{1}\in[0,1)}[G(s-\varepsilon+r_{1}\eta)-G(r_{1}\eta)]>\sup_{r_{2}\in[0,1)}[G(t+\varepsilon+r_{2}\eta)-G(r_{2}\eta)].\label{eq: lem ID part 5 eq 4}
\end{equation}

We now construct $\varepsilon$ that satisfies (\ref{eq: lem ID part 5 eq 4}).
We consider two cases: $s/\eta\in\ZZ$ or $t/\eta\in\ZZ$.

\textbf{Case 1:} suppose $s/\eta\in\ZZ$. Let $k_{1}=s/\eta$. Then
by (\ref{eq: lem ID part 5 eq 3}), we have $G(s+r_{1}\eta)-G(r_{1}\eta)=k_{1}q_{0}$
and $G(t+r_{2}\eta)=G(t-s+r_{2}\eta)+k_{1}q_{0}$. Hence, 
\begin{align*}
 & \inf_{r_{1}\in[0,1)}[G(s+r_{1}\eta)-G(r_{1}\eta)]-\sup_{r_{2}\in[0,1)}[G(t+r_{2}\eta)-G(r_{2}\eta)]\\
 & =k_{1}q_{0}-\sup_{r_{2}\in[0,1)}[G(t-s+r_{2}\eta)+k_{1}q_{0}-G(r_{2}\eta)]\\
 & =\inf_{r_{2}\in[0,1)}[G(r_{2}\eta)-G(r_{2}\eta-(s-t))]\geq\inf_{r_{2}\in[0,1]}[G(r_{2}\eta)-G(r_{2}\eta-(s-t))].
\end{align*}

Since $r_{2}\mapsto G(t+r_{2}\eta)-G(t-s+r_{2}\eta)$ is a continuous
function and $[0,1]$ is a compact set, the above infimum is achieved
at point in $[0,1]$. Since $G$ is strictly increasing and $s-t>0$,
we have that 
\[
\inf_{r_{2}\in[0,1]}[G(r_{2}\eta)-G(r_{2}\eta-(s-t))]>0.
\]

(Otherwise, there would exist $r_{2,*}\in[0,1]$ such that $r_{2,*}\eta=r_{2,*}\eta-(s-t)$.)
Let $\Delta_{1}=\inf_{r_{2}\in[0,1]}[G(r_{2}\eta)-G(r_{2}\eta-(s-t))]$.
We now choose $\varepsilon\in(0,(s-t)/4]$ such that 
\[
\sup_{r_{1}\in[0,1]}[G(r_{1}\eta)-G(r_{1}\eta-\varepsilon)]\leq\Delta_{1}/4
\]
and 
\[
\sup_{r_{2}\in[0,1]}[G(r_{2}\eta-(s-t)+\varepsilon)-G(r_{2}\eta-(s-t))]\leq\Delta_{1}/4
\]

To see that this is possible, notice that $G(\cdot)$ is continuous
and thus is uniformly continuous on the compact set $[-1,\eta]\bigcup[-(s-t),\eta+1-(s-t)]$.
Thus, there exists a constant $\kappa_{1}>0$ such that $\sup_{r_{1}\in[0,1]}[G(r_{1}\eta)-G(r_{1}\eta-\varepsilon)]\leq\kappa_{1}\varepsilon$
and $\sup_{r_{2}\in[0,1]}[G(r_{2}\eta-(s-t)+\varepsilon)-G(r_{2}\eta-(s-t))]\leq\kappa_{1}\varepsilon$.
Hence, we can simply choose $\varepsilon=\min\{(s-t)/4,\Delta_{1}/(4\kappa_{1})\}$.
Therefore, 
\begin{align*}
 & \inf_{r_{1}\in[0,1)}[G(s-\varepsilon+r_{1}\eta)-G(r_{1}\eta)]-\sup_{r_{2}\in[0,1)}[G(t+\varepsilon+r_{2}\eta)-G(r_{2}\eta)]\\
 & =\inf_{r_{1}\in[0,1)}[G(-\varepsilon+r_{1}\eta)+k_{1}q_{0}-G(r_{1}\eta)]-\sup_{r_{2}\in[0,1)}[G(t-s+\varepsilon+r_{2}\eta)+k_{1}q_{0}-G(r_{2}\eta)]\\
 & =\inf_{r_{1}\in[0,1)}[G(-\varepsilon+r_{1}\eta)-G(r_{1}\eta)]-\sup_{r_{2}\in[0,1)}[G(t-s+\varepsilon+r_{2}\eta)-G(r_{2}\eta)]\\
 & =-\sup_{r_{1}\in[0,1)}[G(r_{1}\eta)-G(r_{1}\eta-\varepsilon)]+\inf_{r_{2}\in[0,1)}[G(r_{2}\eta)-G(r_{2}\eta-(s-t)+\varepsilon)]\\
 & \geq-\sup_{r_{1}\in[0,1)}[G(r_{1}\eta)-G(r_{1}\eta-\varepsilon)]+\inf_{r_{2}\in[0,1)}[G(r_{2}\eta)-G(r_{2}\eta-(s-t))]\\
 & \qquad+\inf_{r_{2}\in[0,1)}[G(r_{2}\eta-(s-t))-G(r_{2}\eta-(s-t)+\varepsilon)]\\
 & =-\sup_{r_{1}\in[0,1)}[G(r_{1}\eta)-G(r_{1}\eta-\varepsilon)]+\inf_{r_{2}\in[0,1)}[G(r_{2}\eta)-G(r_{2}\eta-(s-t))]\\
 & \qquad-\sup_{r_{2}\in[0,1)}[G(r_{2}\eta-(s-t)+\varepsilon)-G(r_{2}\eta-(s-t))]\\
 & \overset{\text{(i)}}{\geq}-\Delta_{1}/4+\inf_{r_{2}\in[0,1)}[G(r_{2}\eta)-G(r_{2}\eta-(s-t))]-\Delta_{1}/4\overset{\text{(ii)}}{\geq}-\Delta_{1}/4+\Delta_{1}-\Delta_{1}/4>0,
\end{align*}
where (i) follows by the construction of $\varepsilon$ and (ii) follows
by the definition of $\Delta_{1}$.

\textbf{Case 2:} suppose $t/\eta\in\ZZ$. Let $k_{2}=t/\eta$. Then
by (\ref{eq: lem ID part 5 eq 3}), we have $G(t+r_{2}\eta)-G(r_{2}\eta)=k_{2}q_{0}$
and $G(s+r_{1}\eta)=G(s-t+r_{1}\eta)+k_{2}q_{0}$. Hence, 
\begin{align*}
 & \inf_{r_{1}\in[0,1)}[G(s+r_{1}\eta)-G(r_{1}\eta)]-\sup_{r_{2}\in[0,1)}[G(t+r_{2}\eta)-G(r_{2}\eta)]\\
 & =\inf_{r_{1}\in[0,1)}[G(s-t+r_{1}\eta)+k_{2}q_{0}-G(r_{1}\eta)]-k_{2}q_{0}\\
 & =\inf_{r_{1}\in[0,1)}[G(s-t+r_{1}\eta)-G(r_{1}\eta)]\geq\inf_{r_{1}\in[0,1]}[G(s-t+r_{1}\eta)-G(r_{1}\eta)].
\end{align*}

Since $r_{1}\mapsto G(s-t+r_{1}\eta)-G(r_{1}\eta)$ is a continuous
function and $[0,1]$ is a compact set, the above infimum is achieved
at point in $[0,1]$. Since $G$ is strictly increasing and $s-t>0$,
we have 
\[
\Delta_{2}:=\inf_{r_{1}\in[0,1]}[G(s-t+r_{1}\eta)-G(r_{1}\eta)]>0.
\]

We now choose $\varepsilon\in(0,(s-t)/4]$ such that 
\[
\sup_{r_{2}\in[0,1)}[G(\varepsilon+r_{2}\eta)-G(r_{2}\eta)]\leq\Delta_{2}/4
\]
and 
\[
\sup_{r_{1}\in[0,1]}[G(s-t+r_{1}\eta)-G(s-\varepsilon-t+r_{1}\eta)]\leq\Delta_{2}/4
\]

To see that this is possible, again notice that $G(\cdot)$ is uniformly
continuous on the compact set $[0,\eta+1]\bigcup[(s-t),\eta+1+(s-t)]$.
Thus, there exists a constant $\kappa_{2}>0$ such that $\sup_{r_{2}\in[0,1)}[G(\varepsilon+r_{2}\eta)-G(r_{2}\eta)]\leq\kappa_{2}\varepsilon$
and $\sup_{r_{1}\in[0,1]}[G(s-t+r_{1}\eta)-G(s-\varepsilon-t+r_{1}\eta)]\leq\kappa_{2}\varepsilon$.
Hence, we can simply choose $\varepsilon=\min\{(s-t)/4,\Delta_{2}/(4\kappa_{2})\}$.
Therefore, 
\begin{align*}
 & \inf_{r_{1}\in[0,1)}[G(s-\varepsilon+r_{1}\eta)-G(r_{1}\eta)]-\sup_{r_{2}\in[0,1)}[G(t+\varepsilon+r_{2}\eta)-G(r_{2}\eta)]\\
 & =\inf_{r_{1}\in[0,1)}[G(s-\varepsilon-t+r_{1}\eta)+k_{2}q_{0}-G(r_{1}\eta)]-\sup_{r_{2}\in[0,1)}[G(\varepsilon+r_{2}\eta)+k_{2}q_{0}-G(r_{2}\eta)]\\
 & =\inf_{r_{1}\in[0,1)}[G(s-\varepsilon-t+r_{1}\eta)-G(r_{1}\eta)]-\sup_{r_{2}\in[0,1)}[G(\varepsilon+r_{2}\eta)-G(r_{2}\eta)]\\
 & \geq\inf_{r_{1}\in[0,1)}[G(s-\varepsilon-t+r_{1}\eta)-G(s-t+r_{1}\eta)]+\inf_{r_{1}\in[0,1)}[G(s-t+r_{1}\eta)-G(r_{1}\eta)]\\
 & \quad-\sup_{r_{2}\in[0,1)}[G(\varepsilon+r_{2}\eta)-G(r_{2}\eta)]\\
 & =-\sup_{r_{1}\in[0,1)}[G(s-t+r_{1}\eta)-G(s-\varepsilon-t+r_{1}\eta)]+\inf_{r_{1}\in[0,1)}[G(s-t+r_{1}\eta)-G(r_{1}\eta)]\\
 & \quad-\sup_{r_{2}\in[0,1)}[G(\varepsilon+r_{2}\eta)-G(r_{2}\eta)]\\
 & \overset{\text{(i)}}{\geq}-\Delta_{2}/4+\inf_{r_{1}\in[0,1)}[G(s-t+r_{1}\eta)-G(r_{1}\eta)]-\Delta_{2}/4\overset{\text{(ii)}}{\geq}-\Delta_{2}/4+\Delta_{2}-\Delta_{2}/4>0,
\end{align*}
where (i) follows by the construction of $\varepsilon$ and (ii) follows
by the definition of $\Delta_{2}$.

Thus, we have verified (\ref{eq: lem ID part 5 eq 4}) and thus, $\inf_{\alpha\in\RR}[G(s-\varepsilon+\alpha)-G(\alpha)]>\sup_{\xi\in\RR}[G(t+\varepsilon+\xi)-G(\xi)]$
for some $\varepsilon\in(0,(s-t)/4]$. Clearly $\sup_{\xi\in\RR}[G(t+\varepsilon+\xi)-G(\xi)]\geq G(t+\varepsilon)-G(0)>0$
since $t,\varepsilon>0$ and $G$ is strictly increasing.
\end{proof}
\begin{lem}
\label{lem: ID part 6}Let $\eta>0$. If $\Zcal$ is compact and has
non-empty interior, then there exists an open set $D\subset\Bcal_{+}$
such that for any $d=(d_{1}',d_{2})'\in D$, there exists $z\in\Zcal^{\circ}$
such that $(d_{1}'z+d_{2})/\eta$ is an integer. 
\end{lem}
\begin{proof}
Let $R=\sup_{z\in\Zcal}\|z\|_{2}$. Since $\Zcal$ is compact, $R<\infty$.
Fix any $q=(q_{1}',q_{2})'\in\RR^{K}$ with $q_{2}=\|q_{1}\|_{2}R+1$
and $q_{1}\neq0$. Fix an $z_{0}\in\Zcal^{\circ}$. Then there exists
$\varepsilon>0$ such that $\BB(z_{0},\varepsilon)\subset\Zcal^{\circ}$,
where $\BB(z_{0},\varepsilon):=\{z:\ \|z-z_{0}\|_{2}\leq\varepsilon\}$.
Define $\beta=(\beta_{1}',\beta_{2})'$ with $\beta_{1}=\alpha q_{1}$,
$\beta_{2}=\alpha q_{2}$ and $\alpha=\eta/(z_{0}'q_{1}+q_{2})$.
Clearly, $z'q_{1}+q_{2}\geq-R\|q_{1}\|+q_{2}=1>0$ for any $z\in\Zcal$.
Thus, $\alpha>0$ and $z'\beta_{1}+\beta_{2}\geq\alpha$ for any $z\in\Zcal$.
Hence, $\beta\in\Bcal_{+}$ and $z_{0}'\beta_{1}+\beta_{2}=\eta$. 

Fix any $\delta>0$ such that $\delta\leq\alpha\|q_{1}\|_{2}/(1+\sqrt{\|z_{0}\|_{2}^{2}+1})$
and $\delta\leq\alpha/[2\sqrt{R^{2}+1}]$.

Now consider any $d=(d_{1}',d_{2})'$ such that $\|d-\beta\|_{2}\leq\delta$.
We notice that for any $z\in\Zcal$, $|z'(d_{1}-\beta_{1})+(d_{2}-\beta_{2})|\leq\sqrt{\|z\|_{2}^{2}+1}\|d-\beta\|_{2}\leq\sqrt{R^{2}+1}\delta\leq\alpha/2$.
Thus, for any $z\in\Zcal$, $z'd_{1}+d_{2}\geq z'\beta_{1}+\beta_{2}-\alpha/2\geq\alpha-\alpha/2>0$.
Thus, $d\in\Bcal_{+}$.

We observe that 
\[
\max_{z\in\BB(z_{0},\varepsilon)}d_{1}'z+d_{2}=d_{1}'z_{0}+d_{2}+\varepsilon\max_{\|v\|_{2}\leq1}d_{1}'v=d_{1}'z_{0}+d_{2}+\varepsilon\|d_{1}\|_{2}
\]
and 
\[
\min_{z\in\BB(z_{0},\varepsilon)}d_{1}'z+d_{2}=d_{1}'z_{0}+d_{2}+\varepsilon\min_{\|v\|_{2}\leq1}d_{1}'v=d_{1}'z_{0}+d_{2}-\varepsilon\|d_{1}\|_{2}.
\]

Notice that the maximum and minimum over $\BB(z_{0},\varepsilon)$
are achieved (since $\BB(z_{0},\varepsilon)$ is compact). Hence,
by the intermediate value theorem and the convexity of $\BB(z_{0},\varepsilon)$,
it suffices to show that $d_{1}'z_{0}+d_{2}-\varepsilon\|d_{1}\|_{2}\leq\eta\leq d_{1}'z_{0}+d_{2}+\varepsilon\|d_{1}\|_{2}$.
In other words, it suffices to show that 
\[
|d_{1}'z_{0}+d_{2}-\eta|\leq\varepsilon\|d_{1}\|_{2}.
\]

By $\eta=z_{0}'\beta_{1}+\beta_{2}$, we have $|d_{1}'z_{0}+d_{2}-\eta|=|(d_{1}-\beta_{1})'z_{0}+d_{2}-\beta_{2}|\leq\|(z_{0}',1)\|_{2}\|d-\beta\|_{2}\leq\sqrt{\|z_{0}\|_{2}^{2}+1}\delta$
and $\|d_{1}\|_{2}\geq\|\beta_{1}\|_{2}-\|\beta_{1}-d_{1}\|_{2}\geq\|\beta_{1}\|_{2}-\delta=\alpha\|q_{1}\|_{2}-\delta$.
By construction, $\sqrt{\|z_{0}\|_{2}^{2}+1}\delta\leq\alpha\|q_{1}\|_{2}-\delta$.
The above display holds. Hence, the result holds with $D=\{d:\|d-\beta\|_{2}<\delta\}$. 
\end{proof}
\begin{proof}[\textbf{Proof of Theorem \ref{thm: ID period fun}}]
The first claim follows by Lemma \ref{lem: ID part 3}. The third
claim follows by Lemma \ref{lem: ID part 6} and the second claim.
It remains to prove the second claim. 

We partition $\beta=(\beta_{1}',\beta_{2})'$. By assumption, there
exists $z_{0}\in\Zcal^{\circ}$ such that $k=(z_{0}'\beta_{1}+\beta_{2})/\eta$
is an integer. Clearly $k>0$ since $\beta\in\Bcal_{+}$ (i.e., $z'\beta_{1}+\beta_{2}>0$
for any $z\in\Zcal$). Consider any $b\in\Bcal_{+}$ such that $b=(b_{1}',b_{2})'\neq\beta$.
We proceed in two steps. 

\textbf{Step 1:} show that there exists $z_{*}\in\Zcal^{\circ}$ such
that $(z_{*}',1)\beta\neq(z_{*}',1)b$ and one of $(z_{*}',1)\beta/\eta$
and $(z_{*}',1)b/\eta$ is an integer. 

Suppose that $z_{0}'b_{1}+b_{2}\neq k\eta$. Then we take $z_{*}=z_{0}$
and have $(z_{*}',1)\beta=k\eta\neq(z_{*}',1)b$. Clearly, one of
$(z_{*}',1)\beta/\eta$ and $(z_{*}',1)b/\eta$ is an integer. 

Suppose that $z_{0}'b_{1}+b_{2}=k\eta$. We notice that $\beta_{1}\neq b_{1}$.
To see this, suppose $\beta_{1}=b_{1}$. Then $z_{0}'b_{1}=z_{0}'\beta_{1}$.
Notice that $z_{0}'b_{1}+b_{2}=k\eta=z_{0}'\beta_{1}+\beta_{2}$.
We have $b_{2}=\beta_{2}$. By $\beta_{1}=b_{1}$, this would imply
$\beta=b$, contradicting $\beta\neq b$. Hence, $\beta_{1}\neq b_{1}$
and thus at least one of $\beta_{1}$ and $b_{1}$ is not zero. Then
we only need to discuss three cases.

\uline{Case 1:} suppose $\beta_{1}\neq0$ and $b_{1}\neq0$. Since
$z_{0}\in\Zcal^{\circ}$ and $\beta_{1}\neq0$, $z_{1}=z_{0}+(I-\beta_{1}\beta_{1}/\|\beta_{1}\|_{2}^{2})\Delta\in\Zcal^{\circ}$
when $\|\Delta\|_{2}$ is small enough. We choose $\Delta$ such that
$b_{1}'(I-\beta_{1}\beta_{1}/\|\beta_{1}\|_{2}^{2})\Delta>0$; this
is possible because $b_{1}\neq\beta_{1}$, $\beta_{1}\neq0$ and $b_{1}\neq0$.
Then $z_{1}'\beta_{1}+\beta_{2}=z_{0}'\beta_{1}+\beta_{2}=k\eta$
and $z_{1}'b_{1}+b_{2}=z_{0}'b_{1}+b_{2}+(z_{1}-z_{0})'b_{1}=k\eta+b_{1}'(I-\beta_{1}\beta_{1}/\|\beta_{1}\|_{2}^{2})\Delta>k\eta$.
We now define $t=z_{1}'\beta_{1}+\beta_{2}$ and $s=z_{1}'b_{1}+b_{2}$.
Thus $s=k\eta>t>0$. In other words, we can take $z_{*}=z_{1}$.

\uline{Case 2:} suppose $\beta_{1}\neq0$ and $b_{1}=0$. Since
$z_{0}'b_{1}+b_{2}=k\eta$ and $b_{1}=0$, we have $b_{2}=k\eta$.
Since $\Zcal^{\circ}$ contains an open neighborhood of $z_{0}$ and
$\beta_{1}\neq0$, we can find $z_{1}$ close enough to $z_{0}$ such
that $(z_{1}-z_{0})'\beta_{1}>0$. This means that $z_{1}'\beta_{1}+\beta_{2}=(z_{1}-z_{0})'\beta_{1}+(z_{0}'\beta_{1}+\beta_{2})>k\eta$
and $z_{1}'b_{1}+b_{2}=b_{2}=k\eta$. We now define $s=z_{1}'\beta_{1}+\beta_{2}$
and $t=z_{1}'b_{1}+b_{2}$. Thus, $s>t=k\eta>0$. In other words,
we can take $z_{*}=z_{1}$.

\uline{Case 3:} suppose $\beta_{1}=0$ and $b_{1}\neq0$. The argument
is analogous to the above case. We repeat it here for completeness.
Since $z_{0}'\beta_{1}+\beta_{2}=k\eta$ and $\beta_{1}=0$, we have
$\beta_{2}=k\eta$. Since $\Zcal^{\circ}$ contains an open neighborhood
of $z_{0}$ and $b_{1}\neq0$, we can find $z_{1}$ close enough to
$z_{0}$ such that $(z_{1}-z_{0})'b_{1}>0$. This means that $z_{1}'b_{1}+b_{2}=(z_{1}-z_{0})'b_{1}+(z_{0}'b_{1}+b_{2})>k\eta$
and $z_{1}'\beta_{1}+\beta_{2}=\beta_{2}=k\eta$. We now define $s=z_{1}'b_{1}+b_{2}$
and $t=z_{1}'\beta_{1}+\beta_{2}$. Thus, $s>t=k\eta>0$. In other
words, we can take $z_{*}=z_{1}$.

\textbf{Step 2:} show the final result.

We now prove the result in two cases.

\uline{Case 1:} suppose $(z_{*}',1)\beta>(z_{*}',1)b$. Define
$s=(z_{*}',1)'\beta$ and $t=(z_{*}',1)'b$. Then $s>t>0$ and either
$s/\eta$ or $t/\eta$ is an integer. By Lemma \ref{lem: ID part 5},
there exists $\varepsilon\in(0,(s-t)/4]$ such that 
\[
\inf_{\alpha\in\RR}[G(s-\varepsilon+\alpha)-G(\alpha)]>\sup_{\xi\in\RR}[G(t+\varepsilon+\xi)-G(\xi)]>0.
\]

Let $\BB=\{v\in\RR^{K-1}:\|v\|_{2}<1\}$ be the open unit ball. Notice
that $z_{*}\in\Zcal^{\circ}$. Clearly, there exists $\tau>0$ such
that $z_{*}+\tau\BB\subset\Zcal^{\circ}$, $\tau|v'\beta_{1}|\leq\varepsilon$
and $\tau|v'b_{1}|\leq\varepsilon$ for any $v\in\BB$. Define $\mathcal{Z}_{\tau}=z_{*}+\tau\BB$.
By construction, $\Zcal_{\tau}$ is an open set contained in $\Zcal^{\circ}$.
Notice that for any $z\in\Zcal_{\tau}$, $(z',1)'\beta=(z_{*}',1)'\beta+(z-z_{*})'\beta_{1}\geq(z_{*}',1)'\beta-\varepsilon=s-\varepsilon$
and $(z',1)'b=(z_{*}',1)'b+(z-z_{*})'\beta_{1}\leq(z_{*}',1)'b+\varepsilon=t+\varepsilon$.
Thus, the above display implies that for any $z\in Z_{\tau}$, $(z',1)'\beta\geq s-\varepsilon>t+\varepsilon\geq(z',1)'b$
and thus
\begin{multline*}
\inf_{\alpha\in\RR}[G((z',1)'\beta+\alpha)-G(\alpha)]\geq\inf_{\alpha\in\RR}[G(s-\varepsilon+\alpha)-G(\alpha)]\\
>\sup_{\xi\in\RR}[G(t+\varepsilon+\xi)-G(\xi)]\geq\sup_{\xi\in\RR}[G((z',1)'b+\xi)-G(\xi)]>0.
\end{multline*}

By Lemma \ref{lem: ID part 4}, it follows that for any $z\in\Zcal_{\tau}$,
$\Acal((z',1)'\beta)\bigcap\Acal((z',1)'b)=\emptyset$. 

\uline{Case 2:} suppose $(z_{*}',1)\beta<(z_{*}',1)b$. The argument
is analogous to Case 1 except that we swap the roles of $\beta$ and
$b$. We omit it for simplicity. 

Therefore, we have proved that for any $b\in\Bcal_{+}$ with $b\neq\beta$,
there exists an open set $\Zcal_{b}\subset\Zcal$ such that for any
$z\in\Zcal_{b}$, $\Acal(w'\beta)\bigcap\Acal(w'b)=\emptyset$ with
$w=(z',1)'$. Clearly $\Zcal_{b}\bigcap\Zcal\neq\emptyset$. By Lemma
\ref{lem: topology fact 1}, $P(Z\in\Zcal_{b})>0$. 
\end{proof}

\subsection{Proof of Theorem \ref{thm: neccessity sign}}

In the proof of Theorem \ref{thm: neccessity part 1}, we proved that
any $\beta=(\beta_{1}',\beta_{2})'\in\Bcal_{+}$ is observationally
equivalent to $b=(\beta_{1}',\beta_{2}+\delta)'$ for some $\delta>0$.
Now we pick any $\beta\in\Bcal_{+,0}$. Then $b=(\beta_{1}',\beta_{2}+\delta)'=(\beta_{1}',\delta)\in\Bcal_{+,+}$.
The desired result follows. 

\subsection{Proof of results in Section \ref{sec: test}}
\begin{proof}[\textbf{Proof of Lemma \ref{lem: test}}]
Notice that $\rho(q)=E\oneb\{W'q\geq0\}(Y_{1}-Y_{0})=E\oneb\{W'q\geq0\}\cdot\phi(X)$.
Then by construction $\rho(q)\leq E\oneb\{\phi(X)\geq0\}\cdot\phi(X)$
for any $q$, which means
\[
\sup_{q\in\RR^{K}}\rho(q)\leq E\oneb\{\phi(X)\geq0\}\cdot\phi(X).
\]

On the other hand,
\[
\sup_{q\in\RR^{K}}\rho(q)\geq\rho(\beta)=E\oneb\{W'\beta\geq0\}\cdot\phi(X)\overset{\text{(i)}}{=}E\oneb\{\phi(X)\geq0\}\cdot\phi(X).
\]
where (i) follows by Lemma \ref{lem: manski lem 1}. The above two
displays imply $\sup_{q\in\RR^{K}}\rho(q)=E\oneb\{\phi(X)\geq0\}\cdot\phi(X)=E\max\{\phi(X),0\}$.
The proof of the other claim follows analogously with $\inf_{q}\rho(q)=\rho(-\beta)=E\min\{\phi(X),0\}$. 
\end{proof}
\begin{proof}[\textbf{Proof of Theorem \ref{thm: test}}]
 Define $S_{n}(q)=\sqrt{n}(\hrho_{n}(q)-\rho(q))$. Notice that the
class of mappings $W\mapsto\oneb\{W'q\geq0\}$ indexed by $q$ has
VC-dimension at most $K+2$ (Lemmas 2.6.15 and 2.6.18 of \citet{van1996weak}).
Clearly, the class $\oneb\{W'q\geq0\}(Y_{1}-Y_{0})$ indexed by $q$
is bounded. By Theorem 2.5.2 of \citet{van1996weak}, $S_{n}(\cdot)$
converges weakly to a mean-zero Gaussian process $S_{*}(\cdot)$.
By Theorem 3.6.1 of \citet{van1996weak}, we can approximate the distribution
of $S_{*}$ via the nonparametric bootstrap. 

Of course, we need to check that $S_{*}$ is not degenerate under
$H_{0}$. By Lemma \ref{lem: manski lem 1}, under $H_{0}$, $P(W'\beta\leq0)=1$,
\[
\rho(-\beta)=E\oneb\{W'\beta\leq0\}(Y_{1}-Y_{0})=E(Y_{1}-Y_{0})
\]
and thus
\begin{align*}
E(S_{n}(-\beta))^{2} & =E\oneb\{W'\beta\leq0\}(Y_{1}-Y_{0})^{2}-(\rho(\beta))^{2}\\
 & =E(Y_{1}-Y_{0})^{2}-\left(E(Y_{1}-Y_{0})\right)^{2}={\rm Var}(Y_{1}-Y_{0}).
\end{align*}

Therefore, $\oneb\{W'\beta<0\}(Y_{1}-Y_{0})$ is not degenerate since
$Y_{1}-Y_{0}$ is non-zero variance (due to the assumption of $P(Y_{1}=Y_{0})<1$). 

Let $c_{1-\alpha}^{*}$ be such that $\limsup_{n\rightarrow\infty}P\left(\sup_{q\in\RR^{K}}S_{*}(q)>c_{1-\alpha}^{*}\right)=\alpha$.
If $\tau_{1}\leq-\tau_{2}$, then
\begin{align*}
 & \sqrt{n}\left(\min\{\hat{\tau}_{1},-\hat{\tau}_{2}\}-\min\{\tau_{1},-\tau_{2}\}\right)\\
 & =\sqrt{n}\left(\min\{\hat{\tau}_{1},-\hat{\tau}_{2}\}-\tau_{1}\right)\\
 & =\sqrt{n}\min\{\hat{\tau}_{1}-\tau_{1},\tau_{1}-\hat{\tau}_{2}\}\\
 & \leq\sqrt{n}(\hat{\tau}_{1}-\tau_{1})=\sqrt{n}\left(\sup_{q\in\RR^{K}}\hat{\rho}(q)-\sup_{q\in\RR^{K}}\rho(q)\right)\leq\sqrt{n}\sup_{q\in\RR^{K}}(\hat{\rho}(q)-\rho(q))=\sup_{q\in\RR^{K}}S_{n}(q).
\end{align*}

If $\tau_{1}>-\tau_{2}$, then
\begin{align*}
 & \sqrt{n}\left(\min\{\hat{\tau}_{1},-\hat{\tau}_{2}\}-\min\{\tau_{1},-\tau_{2}\}\right)\\
 & =\sqrt{n}\left(\min\{\hat{\tau}_{1},-\hat{\tau}_{2}\}+\tau_{2}\right)\\
 & =\sqrt{n}\min\{\hat{\tau}_{1}+\tau_{2},\tau_{2}-\hat{\tau}_{2}\}\\
 & \leq\sqrt{n}(\tau_{2}-\hat{\tau}_{2})=\sqrt{n}\left(\inf_{q\in\RR^{K}}\rho(q)-\inf_{q\in\RR^{K}}\hat{\rho}(q)\right)\leq\sqrt{n}\sup_{q\in\RR^{K}}(\rho(q)-\hat{\rho}(q))=\sup_{q\in\RR^{K}}\left(-S_{n}(q)\right).
\end{align*}

Therefore, 
\[
\sqrt{n}\left(\hat{\tau}_{*}-\tau_{*}\right)\leq\left(\sup_{q\in\RR^{K}}S_{n}(q)\right)\cdot\oneb\{\tau_{1}\leq-\tau_{2}\}+\left(\sup_{q\in\RR^{K}}(-S_{n}(q))\right)\cdot\oneb\{\tau_{1}>-\tau_{2}\}.
\]

Therefore, 
\begin{multline*}
P\left(\sqrt{n}\left(\hat{\tau}_{*}-\tau_{*}\right)>c_{1-\alpha}^{*}\right)\\
\leq P\left(\sup_{q\in\RR^{K}}S_{n}(q)>c_{1-\alpha}^{*}\right)\cdot\oneb\{\tau_{1}\leq-\tau_{2}\}+P\left(\sup_{q\in\RR^{K}}(-S_{n}(q))>c_{1-\alpha}^{*}\right)\cdot\oneb\{\tau_{1}>-\tau_{2}\}.
\end{multline*}

Since $S_{n}(\cdot)$ converges weakly to $S_{*}(\cdot)$, $P\left(\sup_{q\in\RR^{K}}S_{n}(q)>c_{1-\alpha}^{*}\right)$
and $P\left(\sup_{q\in\RR^{K}}(-S_{n}(q))>c_{1-\alpha}^{*}\right)$
converge to $\limsup_{n\rightarrow\infty}P\left(\sup_{q\in\RR^{K}}(S_{*}(q))>c_{1-\alpha}^{*}\right)$
and $\limsup_{n\rightarrow\infty}P\left(\sup_{q\in\RR^{K}}(-S_{*}(q))>c_{1-\alpha}^{*}\right)$,
respectively. Both limits are equal to $\alpha$; this is because
$-S_{*}(\cdot)$ and $S_{*}(\cdot)$ have the same distribution. The
result follows since we already proved that $S_{*}(\cdot)$ can be
bootstrapped.%
\end{proof}
\bibliographystyle{apalike}
\bibliography{SC_biblio}

\end{document}